\documentclass{article}
\usepackage{authblk}
\usepackage[utf8]{inputenc}
\usepackage{graphicx}
\usepackage{amsfonts}
\usepackage{amsmath}
\usepackage{nccmath}
\usepackage{amssymb}
\usepackage{fancyhdr}
\usepackage{amsthm}
\usepackage{array}
\usepackage{systeme}
\usepackage{xcolor}
\usepackage{mathtools}

\usepackage{afterpage}
\usepackage{calligra}
\usepackage[T1]{fontenc}
\usepackage{esint}
\usepackage{amsmath}

\usepackage{enumerate}
\usepackage{geometry}
\usepackage{hyperref}
\usepackage{subfig}
\newcommand{\widesim}[2][1.5]{
  \mathrel{\overset{#2}{\scalebox{#1}[1]{$\sim$}}}
}
\usepackage{comment}
\usepackage{tikz}
\usepackage{scalerel}

\title{The Variance Gamma Process for Option Pricing}
\author[1]{Rohan Shenoy}
\author[2]{Peter Kempthorne}
\affil[1]{Department of Mathematics, Imperial College London, rohan.shenoy22@imperial.ac.uk}
\affil[2]{Department of Mathematics, Massachusetts Institute of Technology, kemp2@mit.edu}
\date{August 2024}

\begin{document}

\maketitle

\begin{abstract}
    This paper explores the concept of random-time subordination in modelling stock-price dynamics, and 
We first present results on the Laplace distribution as a Gaussian variance-mixture, in particular a more efficient volatility estimation procedure through the absolute moments. We generalise the Laplace model to characterise the powerful variance gamma model of Madan et al \cite{MadanSeneta} as a Gamma time-subordinated Brownian motion to price European call options via an Esscher transform method. We find that the Variance Gamma model is able to empirically explain excess kurtosis found in log-returns data, rejecting a Black-Scholes assumption in a hypothesis test.\\

\noindent \textbf{Keywords.} Laplace distribution; Gaussian variance-mixture; time-subordinated models; stochastic time-change; variance-gamma process; Esscher transform.
\end{abstract}

\tableofcontents

\newpage

\newtheorem{thm}{Theorem}
\newtheorem{prop}{Proposition}
\newtheorem{lemma}{Lemma}
\newtheorem{definition}{Definition}
\newtheorem{remark}{Remark}

\section{Introduction}\label{Section1}
Under many simple market models, one assumes the returns from an asset follow a log-Normal distribution - the celebrated Black-Scholes option pricing mechanism \cite{blackscholes} is borne out of this assumption for example.\footnote{"Ideal condition" b) of the Black-Scholes formulation: \textit{`The stock price follows a random walk in continuous time with a variance proportional to the square of the stock price. Thus the distribution of possible stock prices at the end of any finite interval is log-Normal. The variance rate of the return on the stock is constant.'}} The literature for Normal (Gaussian) distributions is extensive which has enabled well-grounded research into analytical solutions and closed-form formulas to different problems when an underlying Gaussian assumption is appropriate (often when large sample settings are being considered and the Central Limit Theorem applies). 

However, there are many instances where one cannot reasonably assume that underlying distribution is Gaussian and indeed, in the case of daily log returns data the assumption appears inappropriate.
\begin{figure}[h!]
    \centering
    \subfloat[\centering]{{\includegraphics[width=7cm]{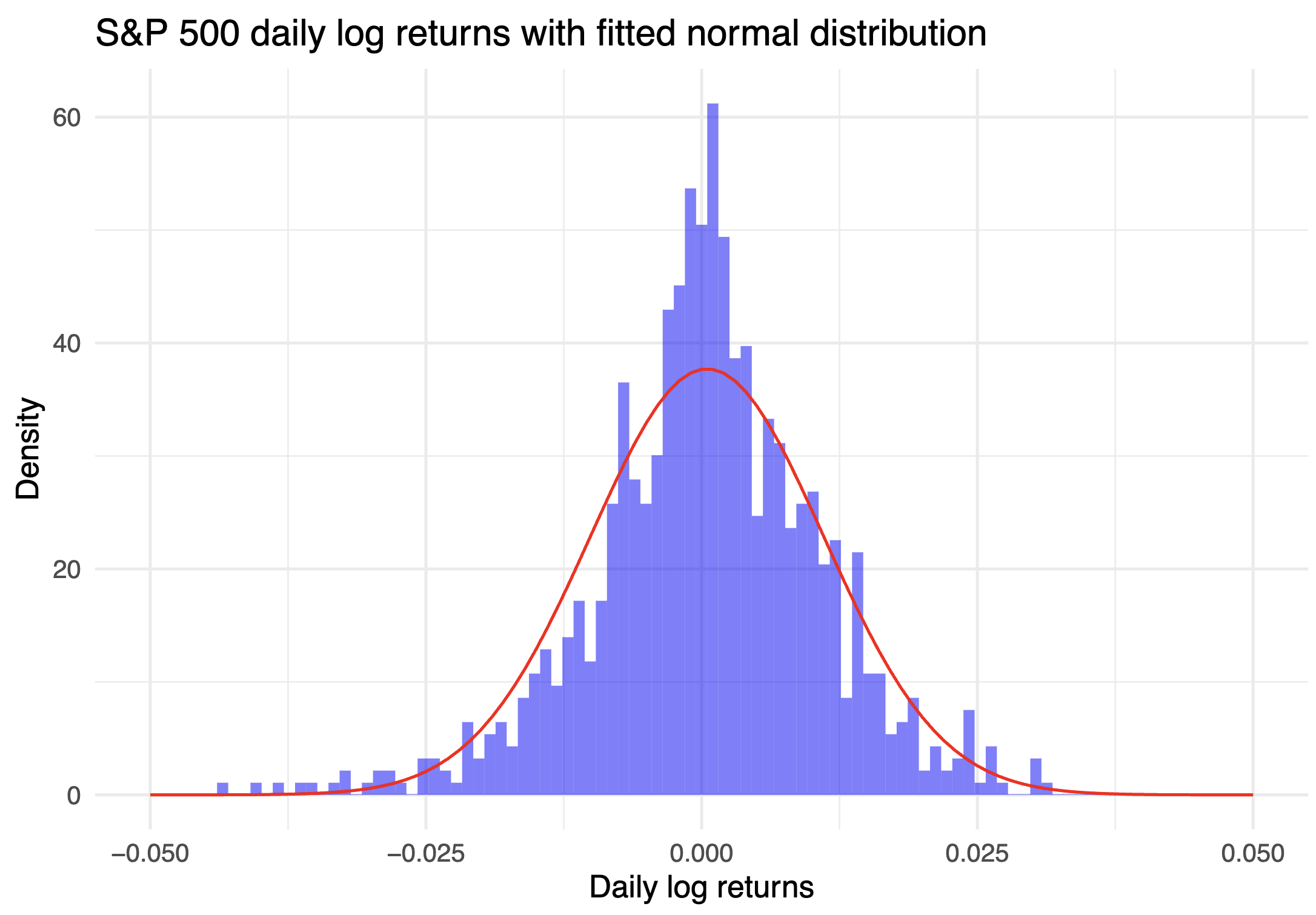} }}
    \qquad
    \subfloat[\centering]{{\includegraphics[width=7cm]{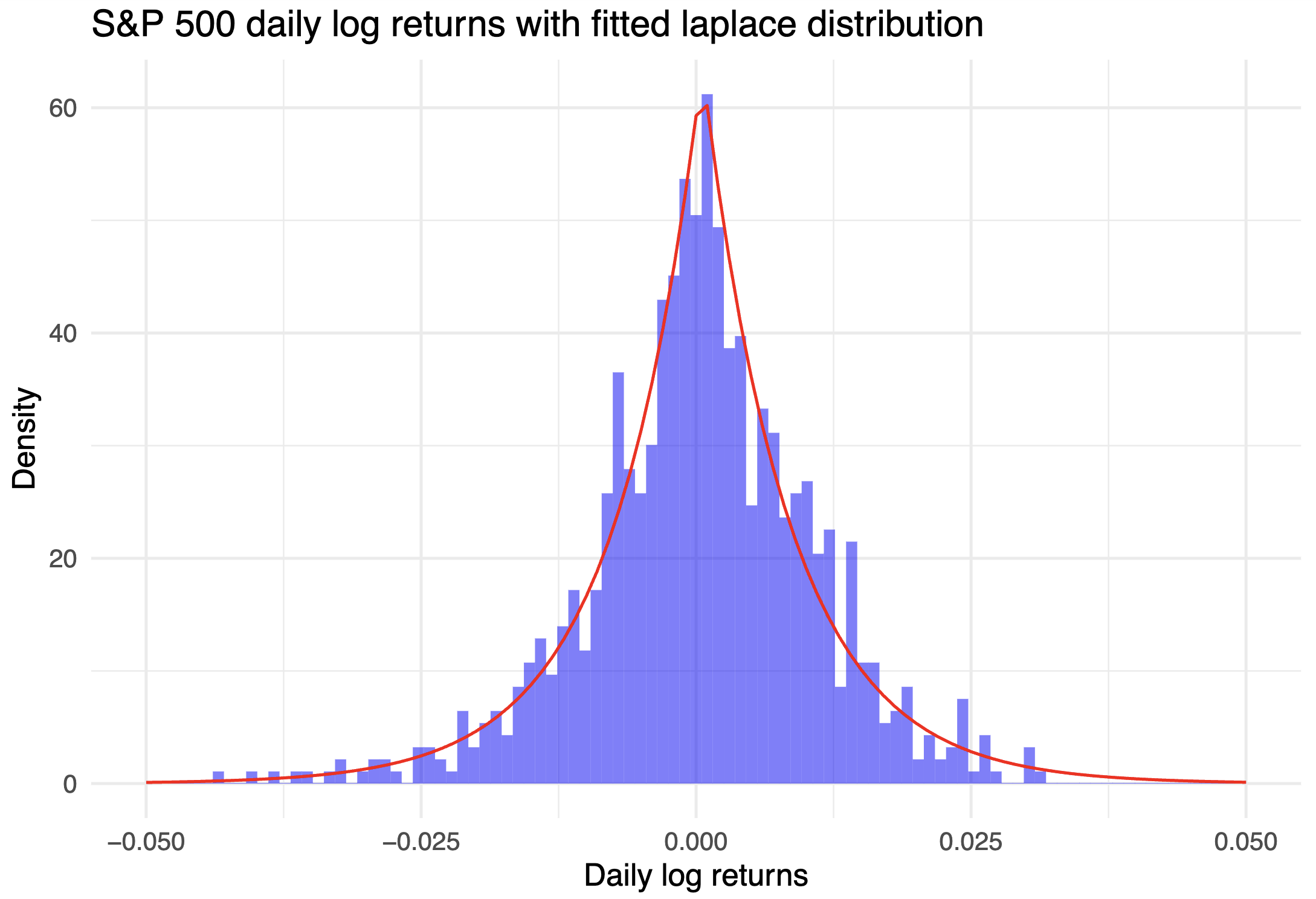} }}
    \caption{Fitted Gaussian vs fitted Laplace for daily log returns from the S$\&$P500 January 2022 to January 2024}
    \label{fig:example}
\end{figure}

The distribution of log returns from the S$\&$P500 index from January 2022 to January 2024 appears to exhibit greater kurtosis than that admitted by a Gaussian. There is a visible difference - the log returns data has heavier tails and a sharper peak at $0$ compared to the fitted Gaussian. Many authors have proposed alternate distributions (such as the Student-t \cite{Praetz}, or Cauchy \cite{KlebanerLandsman}) to the Gaussian to account for these observations.

Here, we first propose the Laplace distribution as an alternative, related to the Gaussian through the generalised Gaussian family (GGD). This GGD family includes an additional shape parameter $\beta$ so that $X\sim GGD(\mu,\sigma,\beta)$ has density
\begin{equation}
    f_X(x) = {\frac {\beta }{2\sigma \Gamma (1/\beta )}}\;e^{-(|x-\mu |/\sigma )^{\beta }}.
\end{equation}
For $\beta = 2$ one recovers the Gaussian family, but $\beta = 1$ gives rise to the Laplace family - a family symmetric about $\mu$, exhibiting sharper peaks and heavier tails which better matches the log returns data. This initially suggests a possible need for alternative estimators of distribution parameters. With a Gaussian assumption, the distribution variance is estimated with the sample variance, the maximum-likelihood estimator with its square root measuring the realized volatility. We will see that sample mean absolute deviation from the sample median provides a more precise volatility estimator when the return distribution satisfies a Laplace assumption.

However, the Laplace distribution may also be expressed as a Gaussian with stochastic variance, and it is this characterisation as a Gaussian variance-mixture which we seek to explore and draw further inference from. Modelling the evolution of the stock price as a geometric Lévy process motivates a generalisation from the Laplace family to the Variance-Gamma family. We explore the Variance Gamma process as a time-subordinated Brownian motion composed as a Brownian motion with Gamma-distributed time increments, as considered in Madan and Milne \cite{MadanMilne}.

It is common to use market-day counts to index price observations - this assumes unit time increments over all successive price observations. Alternatively, the time-subordination characterisation gives weight to the possibility that we can improve on the use of market-day counts, instead implementing a stochastic index. Many authors have proposed possible alternative deterministic indices to market-day counts, such as trading volume or volume of profits; however, it is the works of Madan and Milne (1991) \cite{MadanMilne} and Madan, Carr and Chang (1998) \cite{MadanCarrChang} which have highlighted the applications of a stochastic alternative of this form. 

Under such a model, one infers that the economically relevant time in a market is itself a stochastic process, independent of other random processes affecting market dynamics\footnote{In their paper of 1990, Madan and Seneta comment on this idea, stating: \textit{More informally, one may think of $G(t)$ [the stochastic index] as a formal statement of the remark, "Didn't have much of a year this year," by allowing for an interpretation of how much of a year one actually had.}}. We study the mathematical model of log returns when the time increments follow a Gamma process and the log price increments are attributed to Brownian motion on the stochastic time scale. 

Geometric Brownian motion can be viewed as a sub-model of the Variance Gamma model (where there is no time-subordination) which allows us to observe the benefits of implementing a time-subordination. The Variance Gamma's additional parameters can account for the excess kurtosis observed when modelling the log-return distribution of the S$\&$P500 Index, improving on the Gaussian in this instance. However, we also consider the European option price under a general Variance Gamma model, and compare the pricing performance to that of the Black-Scholes benchmark. It is well-documented that the Black-Scholes formulation leads to known biases - notably the existence of an implied volatility smile upon estimating the volatility parameter. This is often attributed to the observation that log-returns have a risk neutral density with kurtosis above that of a Gaussian; this discrepancy leads to the presence of volatility smiles in the pricing of options under Black-Scholes \cite{BoyleMcDougall}, \cite{PenaRubioSerna}. With the ability to obtain excess kurtosis over the Gaussian, the Variance Gamma density can be expected to correct these biases and propose a causal explanation. Additionally, while the risk-neutral under Black-Scholes model loses the drift term, the Variance Gamma risk neutral density preserves the drift term - it appears the inclusion of this drift term provides significant improvement for option pricing models.

We refer to the work of Madan, Carr and Chang \cite{MadanCarrChang}, who test the performance of both the statistical and risk-neutral $VG$ processes on S$\&$P500 market data, rejecting their respective geometric Brownian motion and Black-Scholes null hypotheses in favour of the $VG$ in the data. The authors are able to demonstrate that the $VG$ model is capable of correcting moneyness biases in pricing errors.

The paper follows the given outline. In section \ref{Section2} we explore the Laplace distribution as an alternative to the Gaussian for fitting daily log returns. We then generalise the Laplace and consider the Variance Gamma distribution in section \ref{Section3}, as well as the corresponding Variance Gamma process. In section 4 we explore the European option pricing mechanism under the Variance Gamma model, obtaining a risk-neutral measure from an Esscher transform method. Section 5 deals with observing the differences between the more general $VG$ model and the nested Black-Scholes model. Section 6 concludes the results and the Appendix details the proofs deferred from the main sections.

\section{The Laplace Distribution}\label{Section2}
\subsection{Initial characterisation}
We first explore the properties of the Laplace distribution as an alternative to the Gaussian. 
\begin{definition}
    A random variable $X$ is said to have a Classical Laplace distribution $CL(\theta, s)$ with mean $\theta\in\mathbb{R}$ and scale $s>0$ if it has real support with density
\begin{equation}
    f_X(x;\theta,s) = \frac{1}{2s}e^{-\frac{|x-\theta|}{s}},\hspace{0.5cm}x\in\mathbb{R}.
\end{equation}
We refer to a standard Classical Laplace random variable for $\theta=0,s=1$.
\end{definition} The graph of this density can be interpreted as a graph of two exponential densities, one for positive deviations and one for negative deviations from the mean $\theta$, both with rate $1/s$ (normalized to integrate to $1$). One computes the characteristic function of a Laplace random variable $X$ from elementary integration (integrating on semi-infinite intervals about $\theta$).
\begin{equation}
    \psi_X(t) = \frac{e^{it\theta}}{1+s^2t^2}, \hspace{0.5cm} t\in\mathbb{R}.
\end{equation}
We can similarly find the moment generating function and it will be useful to write down the Taylor expansion for $\theta=0$ to obtain the central moments of the distribution.
\begin{equation}\label{centralmoments}
    M_X(t) = \dfrac{1}{1-s^2t^2} = \sum_{k=0}^\infty s^{2k}(2k)!\frac{t^{2k}}{(2k)!}, \hspace{0.25cm} -\frac{1}{s} < t < \frac{1}{s}\implies \mathbb{E}\left((X-\theta)^n\right)= 
    \begin{cases}
    0 &\text{$n$ odd}\\
    s^nn! &\text{$n$ even.}
    \end{cases}
\end{equation}
One could alternatively note that the distribution is symmetric so the
odd central moments are $0$, and the absolute value of an $CL(0,s)$ variable is an $\text{Exp}(1/s)$ variable so that the even central moments coincide. So we also find
\begin{equation}\label{centralabsolutemoment}
    \mathbb{E}\left[|X-\theta|^n\right] = s^nn!.
\end{equation}
We may now immediately compute the kurtosis,
\begin{equation}
    \text{Kurt}(X) := \dfrac{\mathbb{E}\left[\left(X-\theta\right)^4\right]}{\left(\mathbb{E}\left[\left(X-\theta\right)^2\right]\right)^2} = \frac{24s^4}{4s^4} = 6.
\end{equation}
\begin{remark}
    The excess kurtosis of a distribution is defined as the additional kurtosis over that of a Gaussian ($=3$) so in this case the excess kurtosis is $3$. For a probability distribution, an excess kurtosis is attributed to greater weight in the tails of the distribution.
\end{remark}

\subsection{Parameter estimation}
We first construct the likelihood function based on a random sample $X_1,\dots,X_n\widesim{i.i.d}CL(\theta,s)$,
\begin{equation}
    \mathcal{L} = f(X_1,\dots,X_n; \theta, s) = \prod_{i=1}^nf(X_i, \theta, s) = \left(\frac{1}{2s}\right)^n\exp\left(-\frac{1}{s}\sum_{i=1}^n|X_i-\theta|\right)
\end{equation}
and obtain the log-likelihood,
\begin{equation}
    \log\mathcal{L} = l(X_1,\dots,X_n; \theta, s) = -n\log(2) -n\log(s) - \frac{1}{s}\sum_{i=1}^n|X_i-\theta|.
\end{equation}
It is evident that we can minimise the sum for any fixed $s$ to find $\hat{\theta}$ and then maximise the expression with respect to $s$ with $\theta$ fixed at $\hat{\theta}$. We begin with minimising the sum to find $\hat{\theta}$.
\begin{lemma}\label{samplemedian}
    For a sample $X_1,\dots,X_n$, the sum $s(\theta) = \sum_{i=1}^n|X_i-\theta|$ is minimised by the sample median. 
\end{lemma}
The proof follows upon ordering the realisations $x_i$ left to right as $x_{k_i}$ and considering the value of the sum when $\theta$ lies in different intervals $[x_{k_i},x_{k_{i+1}}]$. For an odd sample size, the minimising $\theta$ is unique, while for an even sample size this $\theta$ can lie within the closed interval between the middle two values (typically taking the midpoint of these).

We 
conclude that $\hat{\theta}$ is given by the sample median. The variance of this estimator cannot be computed easily, but is given explicitly in Asrabadi \cite{Asrabadi01011985} (eq. 8) where the midpoint median is used for even sample sizes. With the explicit variances, we can generate a plot in Figure (\ref{varfigmed}) of the variances proportional to $s^2/n$. The difference in calculation between odd and even sample sizes causes even sample sizes to be estimated more precisely than closeby odd sample sizes.
\begin{figure}[h!]
    \centering
    \includegraphics[width=0.5\linewidth]{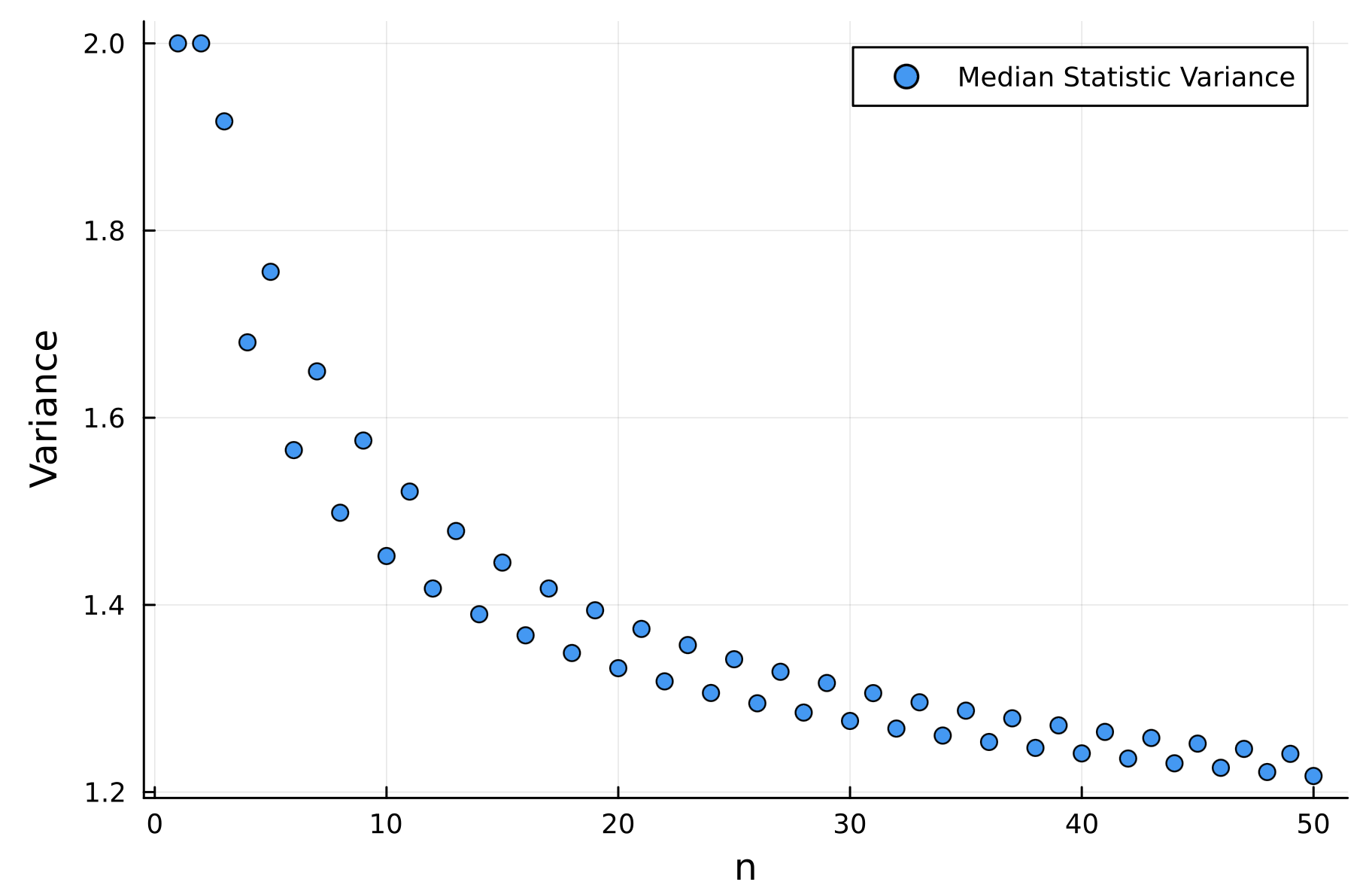}
    \caption{Variance of Laplace sample median proportional to $s^2/n$}
    \label{varfigmed}
\end{figure}
It is shown in Prop 2.6.2 of Kotz et al. \cite{Laplacebook} that the estimator $\hat{\theta}$ is unbiased, consistent and asymptotically Gaussian with $\sqrt{n}(\hat{\theta}-\theta)\overset{d}{\to}N\left(0,s^2\right)$. 

Having found $\hat{\theta}$, we can now maximise the log-likelihood with respect to $s$.

\begin{lemma}\label{thetaestimate}
    Assuming $\theta$ is known, $\hat{s}$ is given by the sample mean absolute deviation from $\theta$ .
    \begin{equation}\label{sigmahat}
        \hat{s} = \frac{1}{n}\sum_{i=1}^n|X_i-\theta|.
    \end{equation}
    Further, the variance of this estimator is equal to $\frac{s^2}{n}$, the Cramér-Rao lower bound for the variance.
\end{lemma}
The proof follows from the standard method, partially differentiating the log-likelihood w.r.t $s$ to obtain $\hat{s}$ and then differentiating again to verify the Cramér-Rao lower bound. We leave the full calculation to the Appendix \ref{thetaestimateappendix} but summarise the findings below.
\begin{prop}
    The maximum likelihood estimator of a random sample from a Classical Laplace $CL(\theta,s)$ variable has an MLE pair given by $\hat{\theta}$, the sample median and $\hat{s} = \frac{1}{n}\sum_{i=1}^n|x_i-\hat{\theta}|$.
\end{prop}

It can be shown (as in Thm 2.6.1 of Kotz et al. \cite{Laplacebook} I) that the pair of estimators $(\hat{\theta},\hat{s})$ for $(\theta,s)$ is consistent, asymptotically Gaussian and efficient (achieving their respective Cramér-Rao lower bounds) with an asymptotic covariance matrix,
\begin{equation}
    \boldsymbol{\Sigma} = \begin{pmatrix}
        s^2 & 0 \\ 0 & s^2
    \end{pmatrix}.
\end{equation} 

If we model log-returns with a Laplace distribution, we should ideally choose this pair of estimators. However, one typically uses a moment estimation method (estimating the sample mean and variance) from the data - these are motivated by their optimality when the log-return distribution is assumed to be Gaussian.

Suppose we instead assume a Laplace distribution but choose to estimate the parameters through moment estimation. We  expect a lack of efficiency and this is indeed case. First, consider estimating $\theta$ by the sample mean.
\begin{lemma}
    Given a sample $X_1,\dots,X_n$, the estimator for $\theta$ given by the sample mean $\Tilde{\theta}= \frac{1}{n}\sum_{i=1}^nX_i$ is unbiased, consistent and asymptotically Gaussian with $\sqrt{n}(\Tilde{\theta}-\theta) \overset{d}{\to}N\left(0,2s^2\right).$
\end{lemma}
\begin{proof}
    The estimator being unbiased follows immediately from the linearity of expectation. Consistency and asymptotic normality are given by the strong law of large numbers and central limit theorem respectively (noting from (\ref{centralmoments}) that the variance of each $X_i$ is indeed $2s^2$).
\end{proof}
We immediately note the relative asymptotic efficiency - the asymptotic variance of the sample mean is twice that of the asymptotic variance of the sample median under the Laplace assumption. Further the variance of the sample mean for finite samples is $2s^2/n$, and while there is no simple closed form for $\hat{\theta}$, the plot above in Figure (\ref{varfigmed}) does confirm that the variance of $\hat{\theta}$ is always less than $\Tilde{\theta}$ for sufficiently large samples.

We now consider the scaled sample standard deviation estimate for $s$.  
\begin{lemma}\label{sigmaestimate}
    The estimator for $s$ given by,
    \begin{equation}
        \Tilde{s} = \sqrt{\frac{1}{2n}\sum_{i=1}^nX_i}
    \end{equation}
    is consistent and asymptotically Gaussian with $\sqrt{n}(\Tilde{s}-s)\overset{d}{\to}N\left(0, \frac{5}{4}s^2\right)$.
\end{lemma}
Consistency follows from the strong law of large numbers while the asymptotic normality follows upon recalling the Laplace central moments from (\ref{centralmoments}) and applying the delta method - the full calculation can be found in the Appendix \ref{sigmaestimateappendix}.

\begin{remark}
    It is important to note that, unlike for the Gaussian, $s$ does not represent the standard deviation of the Laplace distribution (rather the standard deviation of the Laplace is $\sqrt{2}s$). For the Laplace, $s$ represents a scale parameter - we consider estimating $s$ instead of the standard deviation directly to simplify calculations.
\end{remark}

Again note the asymptotic relative efficiency - $\Tilde{s}$ has $5/4$ times greater asymptotic variance than $\hat{s}$ under the Laplace assumption. For the variance of finite samples, there is no simple closed form for $\Tilde{s}$ (as the square root presents complexity in calculations), but the Monte Carlo estimates in Figure (\ref{Laplacesvsim}) do confirm that the variance of $\hat{s}$ is less than $\Tilde{s}$ for suitably large sample sizes \footnote{For smaller sample sizes, the median is more susceptible to leverage from outliers in comparison to the mean, contrasting to the reverse effect for larger sample sizes}.
\begin{figure}[h!]
    \centering
    \includegraphics[width=0.5\linewidth]{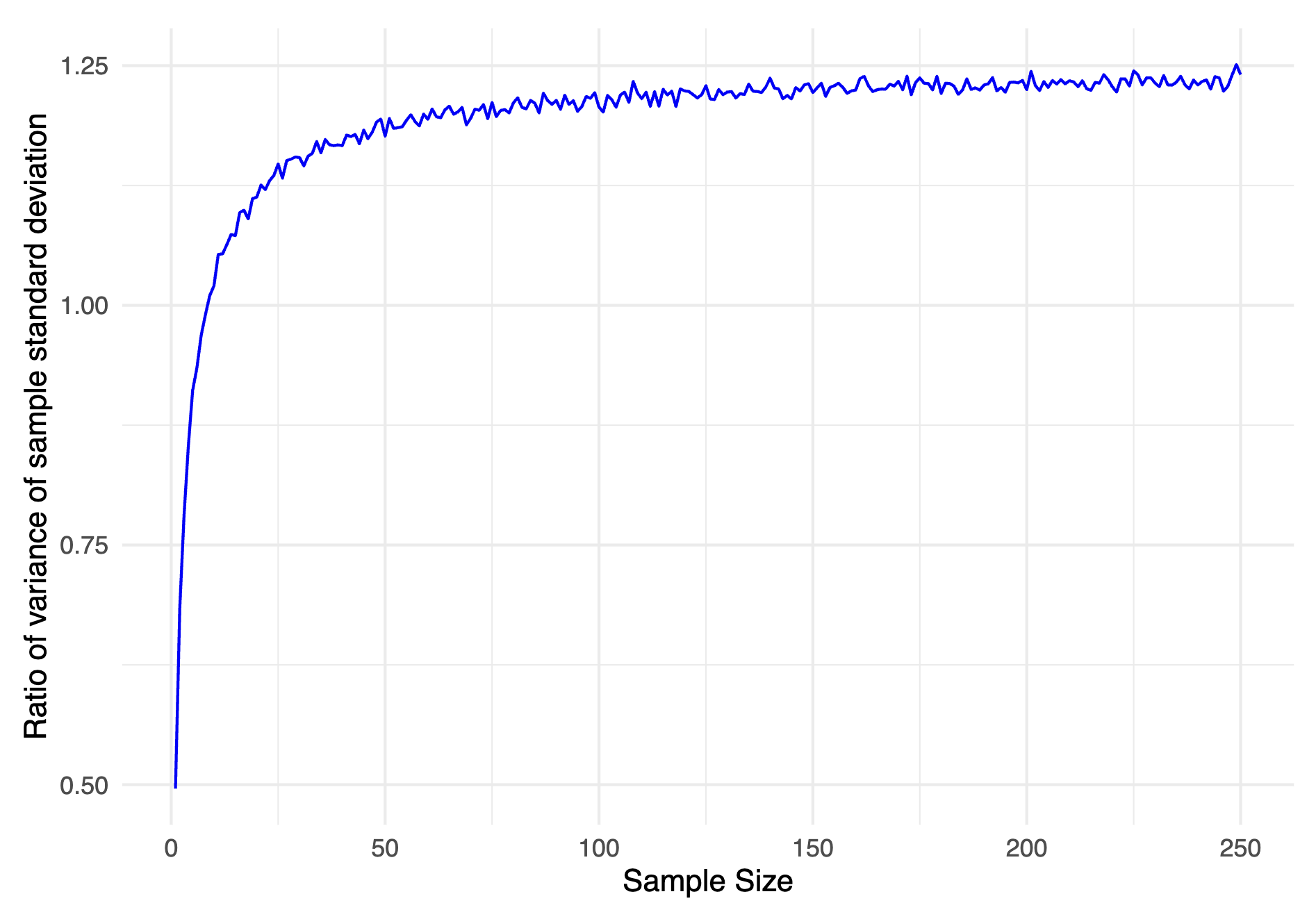}
    \caption{Monte Carlo estimates of the relative efficiency between estimators for the scale parameter $s$}
    \label{Laplacesvsim}
\end{figure}

\begin{remark}\label{remarkestimate}
    If we model the log-return distribution as Laplace instead of Gaussian, we should ideally consider different parameter estimators for our centre and spread of the underlying distribution. Importantly, if we estimate the log-return volatility and assume Laplace, the sample median absolute deviation presents a more efficient estimator for $s$ (which governs the volatility) than the scaled sample standard deviation and represents a more precise statistic for inference and decision analysis (e.g. option pricing) depending on the true parameter (e.g. volatility) of the return distribution.
\end{remark}

\subsection{Laplace distribution as a Gaussian variance-mixture}
We have until now presented the Laplace distribution as an alternative to the Gaussian; however, the distributions are related beyond the generalised Gaussian family. We now explore the Gaussian variance-mixture characterisation from 2.2.1 of Kotz et al. \cite{Laplacebook}.

Suppose that for some known and fixed $\theta$ and $\sigma^2$, $X$ is Gaussian with some mean $\theta$  but with an independent random variance $\sigma^2V$ where $V$ is a non-negative random variable. This is known as a \textit{Gaussian variance-mixture} and may be represented
\begin{equation}\label{normmix}
    X = \theta + \sigma\sqrt{V}Z
\end{equation}
where $Z\sim N(0,1)$, and $V$ is a random variable independent of $Z$. Intuitively one understands this by conditioning on $V$, re-scaling the standard Gaussian to determine the variance, and ultimately shifting the mean. Suppose we now centralise $X$ with $\theta=0$, specify $V$ to be a standard exponential random variable with rate $1$, and let $\sigma = \sqrt{2}$ so that $X=\sqrt{2V}Z$.
We will see that leads to a standard Laplace distribution.
\begin{prop}\label{Laprep}
    Let $Z$ have a standard Gaussian and $V$ be an independent a standard exponential random variable. Then the random variable $X = \sqrt{2V}Z$ has a standard Laplace distribution.
\end{prop}
\begin{remark}\label{conditioning}
    Before providing the proof, we reflect on the Gaussian variance-mixture characterisation of $X$. When encountering such a mixture, we can first condition on the random variance to reduce the setup to Normality - something we understand well - then deal with the random variance separately. We will return to this idea throughout.
\end{remark}

\begin{proof}
    Here, we appeal to the identity theorem for the characteristic functions of probability distributions. For $V$ a standard exponential with density $f_V(v) = e^{-v}$, we have $M_V(t) = ({1-t})^{-1}$. For $Z$ a standard Gaussian variable, we have $\psi_Z(t) = e^{-t^2/2}$. With these in hand, we calculate the characteristic function of $X$ by conditioning on the value of the (non-negative) variance $V$,
    \begin{align*}
        \mathbb{E}\left[e^{itX}\right] &= \mathbb{E}\left[e^{it\sqrt{2V}Z}\right]=\mathbb{E}_V\left[\mathbb{E}_Z\left[e^{it\sqrt{2V}Z}\Big\vert V\right]\right]=\mathbb{E}_V\left[\psi_X\left(t\sqrt{2V}\right)\right]\\
        &=\mathbb{E}_V\left[e^{-t^2V}\right]=M_V\left(-t^2\right)=\frac{1}{1+t^2},
    \end{align*}
    which indeed coincides with the characteristic function of the standard Laplace random variable.
\end{proof}
If we generalise to $X=\theta + \sigma\sqrt{V}Z$, for fixed $\theta$ and $\sigma$, $Z$ a standard Gaussian r.v but $V\sim \exp(\lambda)$ one employs a similar method to show that $X\sim CL(\theta, \sigma/\sqrt{2\lambda})$. We aim to generalise this Gaussian variance-mixture further in section \ref{Section3}.

\begin{remark}
    This is a particularly useful characterisation of the Laplace distribution, as we will see in later sections. However, there are alternate characterisations of the Laplace distribution which provide a different interpretation if one models a log-return distribution with a Laplace distribution. 
    
    For example, the distribution can also be represented as a difference of two i.i.d. Exponential distributions (this can be seen by splitting the characteristic function $(1+t^2)^{-1} = (1+it)(1-it)$ for the standard Laplace, and similar for general Laplace). One can find further implications of considering the general Classical Laplace and later the variance gamma model as this difference in Fischer et al. \cite{fischer2023variancegammadistributionreview}.

    One can also characterise the Laplace as a maximum entropy distribution. The maximum entropy principle states that considering all distributions satisfying certain constraints, one should select the one with the largest entropy as characterising this distribution maximises randomness subject to the constraint - this idea is explored in detail by Rényi \cite{RenyiEntropy}. As a result, finding a maximum entropy distribution presents a suitable procedure for robust inference, detailed in Jaynes \cite{Jaynes}. If one specifies the first absolute central moment, the maximum entropy family with real support is the corresponding Laplace family, shown in the Appendix \ref{Entropyappendix}).

    For our study, we seek to explore inference from the random variance representation which allows us to consider a stochastic time change for a time-subordinated Brownian motion, and the corresponding interpretation for the model.
\end{remark}

\section{The Variance Gamma Process}\label{Section3}
\subsection{The symmetric Variance Gamma distribution}\label{sectionSVGd}

If one assumes the log-return distribution for an asset is Gaussian, geometric Brownian motion arises as the associated Lévy process to describe the evolution of the stock price. The fact that the sum of independent Gaussians is still Gaussian greatly helps in the characterisation of Brownian motion as a Lévy process with Gaussian distributed increments. If we consider an equivalent under a Laplace assumption, we need to understand sums of Laplace distributions and in doing this we will see that weakening the Laplace assumption to be slightly more general aids the description of the associated Lévy process.

Suppose that $X_i=\sqrt{V_i}Z_i$ for $i=1,\dots, n$ where $V_i\widesim{i.i.d}\exp(\lambda)$ and $Z_i\widesim{i.i.d}N(0,1)$ are independent of the $V_i$'s. From the generalisation of Prop. (\ref{Laprep}),
\begin{equation}
    \psi_{X_i}(t) = \left({1+\frac{t^2}{2\lambda}}\right)^{-1},
\end{equation}
so that for $S_n := \sum_{i=1}^nX_i$,
\begin{equation}\label{Ssum}
    \psi_{S_n}(t) = \left({1+\frac{t^2}{2\lambda}}\right)^{-n}.
\end{equation}
\begin{remark}\label{SGamma}
    We recognise a strong resemblance with the characteristic function of a Gamma distribution (given below) - one with shape $n$ and rate $\lambda$, only here the argument $t$ has been replaced with $-t^2/2$. We saw this arise in a similar manner with the Exponential characteristic function in the proof of Prop. (\ref{Laprep}) and this hints at a nice form for the more general Gaussian variance-mixture.
    
    Intuitively one might expect a connection with the Gamma distribution - as the variance of a sum of independent Gaussians is the sum of their variances (which are Exponentially distributed here), and a sum of independent Exponential distributions is a Gamma distribution. We formalise this notion below.
\end{remark}
\begin{definition}
Let $X$ be a Gamma distributed random variable with shape $\alpha>0$, rate $\beta>0$, denoted $X\sim\text{Gamma}(\alpha,\beta)$. Then the density $g_X(x)$ for $X$ is given by
\begin{equation}
    g_X(x) = \dfrac{\beta^\alpha}{\Gamma(\alpha)}x^{\alpha-1}e^{-\beta x},\hspace{0.5cm} x>0.
\end{equation}
\end{definition}
The characteristic and moment-generating functions are obtained from integration (employing the integral definition of the Gamma function).
\begin{equation}
    \psi_{X}(t) = \left(1-{\frac{it}{\beta}}\right)^{-\alpha }, \hspace{0.2cm} M_X(t) = \left(1-{\frac {t}{\beta }}\right)^{-\alpha }.
\end{equation}
We are now in place to define the Symmetric Variance Gamma distribution as a Gaussian variance-mixture.
\begin{definition}
    Let $X = \theta + \sigma\sqrt{V}Z$ where $\theta\in\mathbb{R}, \sigma>0$ are fixed, $Z$ has a standard Gaussian, and $V$ is an independent $\text{Gamma}(\alpha,\beta)$ distribution. Then $X\sim SVG(\theta,\sigma, \alpha, \beta)$ is said to have a Symmetric Variance Gamma distribution.
\end{definition}
\begin{remark}\label{SVGdefremark}
    Under this definition $X$ is a symmetric distribution about $\theta$, and so $\mathbb{E}[X] = \theta$. Importantly, one observes (e.g. from the characteristic function below) that $\sigma$ and $\beta$ are not-identifiable. However, this actually provides flexibility in interpreting properties of the distribution across different settings.
    
    If we restrict $V$ to have expectation $\mu=\alpha/\beta = 1$, then $\sigma$ becomes a volatility parameter. This can be seen by using the independence of $V$ and $Z$,
    \begin{equation}
        \text{Var }[X] = \mathbb{E}\left[\left(\sigma\sqrt{V}Z\right)^2\right] = \sigma^2\mathbb{E}\left[{V}Z^2\right] = \sigma^2\mathbb{E}[V]\mathbb{E}\left[Z^2\right] = \sigma^2.
    \end{equation}
    This is a particularly useful notion which we will employ in section \ref{SectionVGProcess} to define the Variance Gamma process. 
    
    Alternatively, we could fix $\sigma = 1$ and directly observe changes to $V$ across different settings e.g. over different time periods in the Variance Gamma process. This is another valuable notion one can use for modelling stochastic volatility.
\end{remark}
For $\alpha = 1$, we recover a Laplace distribution as the Gamma-distributed variance is simply Exponential. However, in general there is no closed form for the density of the $SVG$ distribution, and so we write through a Radon-Nikodym derivative (effectively conditioning on the Gamma-distributed variance). For a standard Gaussian density $\phi(x)$ and $\text{Gamma}(\alpha,\beta)$ density $g(x)$, one writes
\begin{equation}\label{VG density}
    f_X(x)=\int_{0}^{\infty}\phi\left(\dfrac{x-\theta}{\sigma\sqrt{v}}\right)g(v)dv=\int_{0}^{\infty}\frac{1}{\sqrt{2\pi\sigma^2v}}e^{-(x-\theta)^2/{(2\sigma^2v)}}\dfrac{\beta^\alpha}{\Gamma(\alpha)}v^{\alpha-1}e^{-\beta v}dv.
\end{equation}
As before, it will be useful to find the characteristic function of $X$ for future analysis. It has a generalised form of the Laplace characteristic function, as hinted in Remark (\ref{SGamma}).
\begin{prop}\label{VGcf}
    Let $Z$ be a standard Gaussian and $V$ be an independent $\text{Gamma}(\alpha,\beta)$ variable. In addition, let $\sigma>0$ be known and fixed. Then the characteristic function of the Symmetric $VG$ variable $X=\sigma \sqrt{V}Z$ is given by
    \begin{equation}
        \psi_{X}(t) = \left(1+\dfrac{t^2\sigma^2}{2\beta}\right)^{-\alpha}.
    \end{equation}
\end{prop}
\begin{proof}
    As before, we will show this by conditioning on the value of the variance $V$. We first note that for $V\sim\Gamma(\alpha,\beta)$ we have $M_V(t)=({1-{t}/{\beta}})^{-\alpha}$ and again for $Z$, a standard Gaussian variable we have $\psi_Z(t) = e^{-t^2/2}$. Then proceeding as before, 
    \begin{align*}
        \mathbb{E}\left[e^{itX}\right] &= \mathbb{E}\left[e^{it\sigma\sqrt{V}Z}\right]=\mathbb{E}_V\left[\mathbb{E}_Z\left[e^{it\sigma\sqrt{V}Z}\hspace{0.1cm}\Big\vert V\right]\right]=\mathbb{E}_V\left[\psi_Z\left(t\sigma\sqrt{V}\right)\right]=\mathbb{E}_V\left[e^{-t^2\sigma^2V/2}\right]\\
        &=M_V\left(-t^2\sigma^2/2\right)=\left(1+\dfrac{t^2\sigma^2}{2\beta}\right)^{-\alpha}.
    \end{align*}
\end{proof}

\subsection{The Variance Gamma distribution}\label{3.2}

Having motivated and explored a generalisation from the symmetric Laplace distribution to the symmetric Variance Gamma, we now consider the full Variance Gamma distribution which can be represented as a Gaussian with both stochastic variance \textit{and} stochastic mean. This generalisation enables the inclusion of a drift term in the associated Lévy process, accounting for possible asymmetry in the statistical and risk-neutral log return distributions.

\begin{definition}
Let $X = c + \theta V + \sigma\sqrt{V}Z$ where $c,\theta\in\mathbb{R}, \sigma>0$ are fixed, $Z$ has a standard Gaussian, and $V$ is an independent $\text{Gamma}(\alpha,\beta)$ distribution. Then $X\sim VG(c,\theta,\sigma, \alpha, \beta)$ is said to have a Variance Gamma distribution.
\end{definition}
\begin{remark}\label{reparametrisegamma}
    Under this new definition if $\theta\neq0$, $X$ is no longer a symmetric distribution due to the $\theta V$ term. For $\theta\neq0$, $\sigma$ no longer represents the volatility as $\text{Var}[X] = \theta^2\text{Var}[V] + \sigma^2\mathbb{E}[V]$. However, as in Remark \ref{SVGdefremark}, if we let $\mathbb{E}[V] = \mu = 1$ and denote $\text{Var}[V] = \nu$ then the volatility is given by $\theta^2\nu + \sigma^2$. 
    
    This motivates a reparametrisation which we will use to define the Gamma Process in section \ref{SectionVGProcess}. We can alternatively define a Gamma-distributed random variable $X\sim \text{Gamma}(\alpha, \beta)$ in terms of its positive mean $\mu = \alpha/\beta$ and variance $\nu = \alpha/\beta^2$ so that under a re-parametrisation, we write $X\sim\Gamma_{\mu,\nu}$ with
\begin{equation}
    f_X(x) = \dfrac{\left(\frac{\mu}{\nu}\right)^{\mu^2/\nu}}{\Gamma\left(\mu^2/\nu\right)}x^{\mu^2/\nu-1}e^{-\mu/\nu x}, \hspace{0.5cm} \psi_{X}(t) = \left(\dfrac{1}{1-\frac{i\nu}{\mu}t}\right)^{\mu^2/\nu}, \hspace{0.5cm} M_{X}(t) = \left(\dfrac{1}{1-\frac{\nu}{\mu}t}\right)^{\mu^2/\nu}.
\end{equation}
\end{remark}
We may again write the density of a $VG$ variable in terms of a Radon-Nikodym density, similar to (\ref{VG density}), using the original parametrisation in terms of $\text{Gamma}(\alpha,\beta)$,
\begin{equation}\label{General VG Density}
    f_X(x) = \int_{0}^{\infty}\phi\left(\dfrac{x-c-\theta v}{\sigma\sqrt{v}}\right)g(v)dv=\int_{0}^{\infty}\frac{1}{\sqrt{2\pi\sigma^2v}}e^{-(x-c-\theta v)^2/{(2\sigma^2v)}}\dfrac{\beta^\alpha}{\Gamma(\alpha)}v^{\alpha-1}e^{-\beta v}dv,
\end{equation}
and compute the characteristic function in a similar manner to Prop. (\ref{VGcf}).
\begin{prop}\label{aVGcf}
    Let $Z$ be a standard Gaussian and $V$ be an independent $\text{Gamma}(\alpha,\beta)$ variable. Further, let $\theta\in\mathbb{R},\sigma>0$ be known and fixed. Then the characteristic function of the $VG$ variable $X=\theta V + \sigma \sqrt{V}Z$ is given by
    \begin{equation}
        \psi_{X}(t) = \left(1 - \frac{i\theta t}{\beta} + \dfrac{t^2\sigma^2}{2\beta}\right)^{-\alpha}.
    \end{equation}
\end{prop}
\begin{proof}
   Proceeding as in the proof of Prop. (\ref{VGcf}) by conditioning on the value of $V$, we find
   \begin{align}
       \mathbb{E}\left[e^{itX}\right] &= \mathbb{E}\left[e^{it(\theta V+\sigma\sqrt{V}Z)}\right]=\mathbb{E}_V\left[\mathbb{E}_Z\left[e^{it(\theta V+\sigma\sqrt{V}Z)}\hspace{0.1cm}\Big\vert V\right]\right] = \mathbb{E}_V\left[e^{it\theta V}\psi_Z\left(t\sigma\sqrt{V}\right)\right] \\
       &=\mathbb{E}_V\left[e^{it\theta V- t^2\sigma^2V/2}\right] = \mathbb{E}_V\left[e^{(it\theta- t^2\sigma^2/2)V}\right]\\
       &=\int_{0}^\infty e^{(it\theta- t^2\sigma^2/2)v}\dfrac{\beta^\alpha}{\Gamma(\alpha)}v^{\alpha-1}e^{-\beta v}dv = \int_{0}^\infty \dfrac{\beta^\alpha}{\Gamma(\alpha)}v^{\alpha-1}e^{-(\beta-it\theta+t^2\sigma^2/2)v} dv,
   \end{align}
   from which the result follows upon applying the definition of the Gamma function (noting $\beta+t^2\sigma^2/2>0$),
   \begin{equation}
       \int_{0}^\infty\dfrac{\beta^\alpha}{\Gamma(\alpha)}v^{\alpha-1}e^{-(\beta-it\theta+t^2\sigma^2/2)v} dv = \dfrac{\beta^\alpha}{\Gamma(\alpha)}\dfrac{\Gamma(\alpha)}{\left(\beta-it\theta+t^2\sigma^2/2\right)^\alpha} = \left(1 - \frac{i\theta t}{\beta} + \dfrac{t^2\sigma^2}{2\beta}\right)^{-\alpha}
   \end{equation}
   as desired.
\end{proof}

\subsection{The Variance Gamma process}\label{SectionVGProcess}
We have characterised the $VG$ distribution and now turn to the associated $VG$ process. This will follow a similar analysis, considering the Lévy process as a composition mixture of a Gaussian process (a Brownian motion) and a Gamma process as described in 4.2 of Kotz et al. \cite{Laplacebook}. First, we clarify our definition of a Brownian motion.
\begin{definition}\label{Brownian}
    Let $(W_t)_{t\geq0}$ be a standard Wiener process where $W_0 = 0$ a.s. , $W_t$ has stationary Gaussian increments so that $W_{t+h} - W_t \sim N(0, h)$ where the distributions are independent over non-overlapping time intervals, and $W_t$ is continuous $a.s.$ . Then the process $(b_t)_{t\geq0}$ given by
    \begin{equation}\label{Brownian motion}
        b_t = b(t;\theta, \sigma) = \theta t + \sigma W_t
    \end{equation}
    is a Brownian motion with drift $\theta\in\mathbb{R}$ and volatility $\sigma>0$.
\end{definition}
\noindent It is useful to recall the characteristic function for a Brownian motion,
\begin{equation}\label{bmcf}
    \psi_{b(t;\theta,\sigma)}(u) = \exp\left(i\theta tu -\frac{1}{2}\sigma^2t u^2 \right).
\end{equation}
Our $VG$ distribution was characterised as a Gaussian with a random Gamma distributed mean and variance. In a Brownian motion, the increments have a mean and variance proportional to the size of the time increment, so to create an analogue we seek a way of transforming the time increments randomly using a Gamma distribution and then evaluating the Brownian motion at the resulting transformed time - this is known as a \textit{time-subordinated} Brownian motion where the Gamma process is the time-subordinator. We define this below (employing the re-parametrisation of remark \ref{reparametrisegamma}).
\begin{definition}\label{Gamprocess}
    Let $(\gamma_t)_{t\geq0} = \gamma(t;\mu,\nu)$ be a Gamma process for $\mu\in\mathbb{R}, \nu\geq0$.\footnote{The degenerate case $\nu=0$ is included and viewed as a deterministic limit process $\mu t$ as $\nu\to0.$} Then $\gamma_0 = 0$ $a.s.$ , $\gamma_t$ has stationary Gamma distributed increments s.t. $\gamma_{t+h}-\gamma_t\sim \Gamma_{\mu h,\nu h}$ where the distributions are independent over non-overlapping time intervals, and $\gamma_t$ is continuous a.s. .
\end{definition}
\noindent Again it will be useful to have the characteristic and moment generating functions,
\begin{equation}\label{gammatcf}
    \psi_{\gamma(t;\mu,\nu)}(u) = \left(\dfrac{1}{1-\frac{i\nu}{\mu}u}\right)^{\mu^2t/\nu}, \hspace{0.5cm} M_{\gamma(t;\mu,\nu)}(u) = \left(\dfrac{1}{1-\frac{\nu}{\mu}u}\right)^{\mu^2t/\nu}
\end{equation}
Now let $\mu = 1$ so that the mean rate of the Gamma process is $1$. This is a helpful characterisation as it means that at time $T$, the expectation of our random Gamma distributed time $\mathbb{E}\left[\gamma(T;1,\nu)\right]=T$ and we have random fluctuation only about the deterministic time.

\begin{remark}\label{gammatime}
    For our purposes, we may think of the Gamma process as a map $(T_t)_{t\geq0}\mapsto(\gamma_t)_{t\geq0}$ from the `deterministic time clock' $T_t=t$ onto a new `Gamma time clock' $\gamma_t$.
    
    Figure \ref{gammasim} provides examples for some of the possible properties of the Gamma subordinator process. We can think of the original clock being slowed in regions of low gradient and sped up up for high gradient. The markers correspond to unit time intervals on the new clock, and demonstrate instances of this time dilation - regions in which the markers are more spread out correspond to time being dilated and slowed, whereas regions with clusters of markers correspond to time being contracted and sped up. 
    
    The mean rate of the time change is null, in that on average 1 unit of deterministic time corresponds to 1 unit of Gamma-changed time (as shown by the $y=x$ line).
\end{remark}
\begin{figure}[h!]
    \centering
    \includegraphics[width=0.45\linewidth]{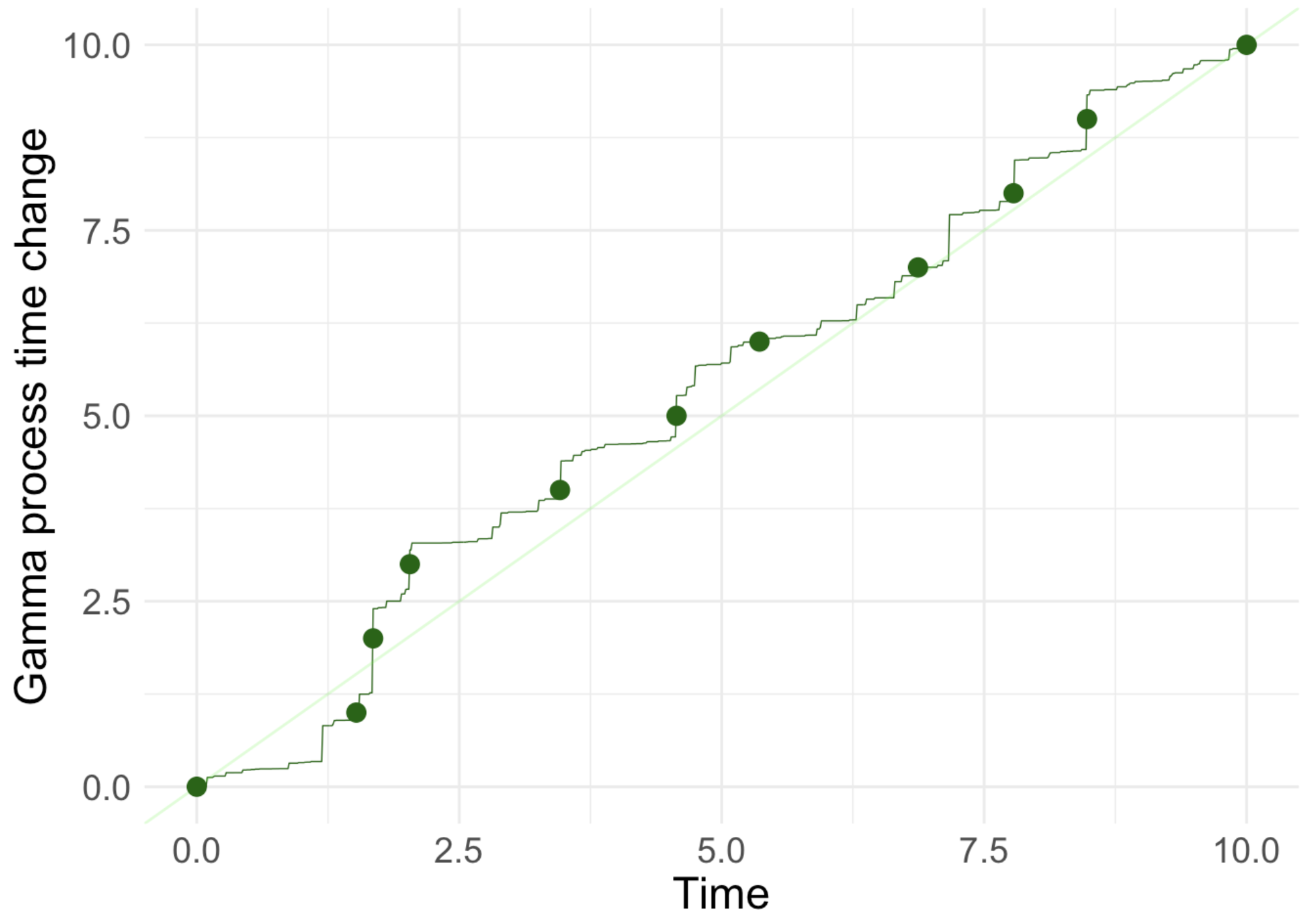}
    \caption{Simulated Gamma process: $\gamma(t;\mu =1,\nu = 0.3)$}
    \label{gammasim}
\end{figure}

\noindent We can finally characterise our $VG$ process as Gamma subordinated Brownian motion.
\begin{definition}
    Let $(b_t)_{t\geq0} = b(t;\theta, \sigma)$ be a Brownian motion as in definition (\ref{Brownian}),  and let $(\gamma_t)_{t\geq0} = \gamma(t;1,\nu)$ be a Gamma process as in definition (\ref{Gamprocess}). Then the time-subordinated process $(X_t)_{t\geq0}$ given by the composition
    \begin{align}
        X_t = X(t;\theta,\sigma, \nu) &= b\big(\gamma(t;1,\nu);\theta, \sigma\big) \\
        &= \theta \gamma_t + \sigma W_{\gamma_t},\label{thetagammasigmawienergamma}
    \end{align}
    is known as a Variance Gamma process denoted $VG(t;\theta,\sigma, \nu)$.
\end{definition}
\begin{remark}
    Breaking down this definition, there are two sources of randomness. First is the Wiener process which acts the same as in the Brownian motion, creating Gaussian noise about the drift $\theta t$. However, there is also the Gamma time-subordination, so that we think of the process as first randomly mapping deterministic time onto a Gamma-changed time as in remark \ref{gammatime} and then evaluating the Brownian motion at the Gamma time change \footnote{We have chosen the time-changed Brownian motion as our characterisation of the $VG$ process. Similar to the comments in section 2.4.2, there are alternate characterisations of the $VG$ process e.g. as a difference of independent Gamma processes or as the limit of a compound poisson approximation. We do not cover these here, as we seek to explore inferences based on the random time-change characterisation. The alternate characterisations are detailed 4.2.3 of Kotz et al. \cite{Laplacebook}.}. For a more detailed review of time-subordinated processes, see section 1.3 of Applebaum \cite{Applebaum}.
    
    A Brownian motion is simulated in Figure \ref{VGcomponents}(a) and mapped onto the VG process in Figure \ref{VGcomponents}(b) through the Gamma time-subordination of Figure \ref{gammasim}. I.e. $b(t)\mapsto b(\gamma(t))$. Through the time-transformation, the markers in the Brownian motion are mapped onto the VG process in the Figure via $(t,b(t))\mapsto (\gamma(t),b(\gamma(t))$. 
    
    These demonstrate the effect of the time-subordination - while the general shape of the motion is mostly preserved, one observes the dilation and contraction of the Brownian motion along the time-axis caused by the time-subordination.
\end{remark}
\begin{figure}[h!]
    \centering
    \subfloat[\centering]{{\includegraphics[width=0.45\linewidth]{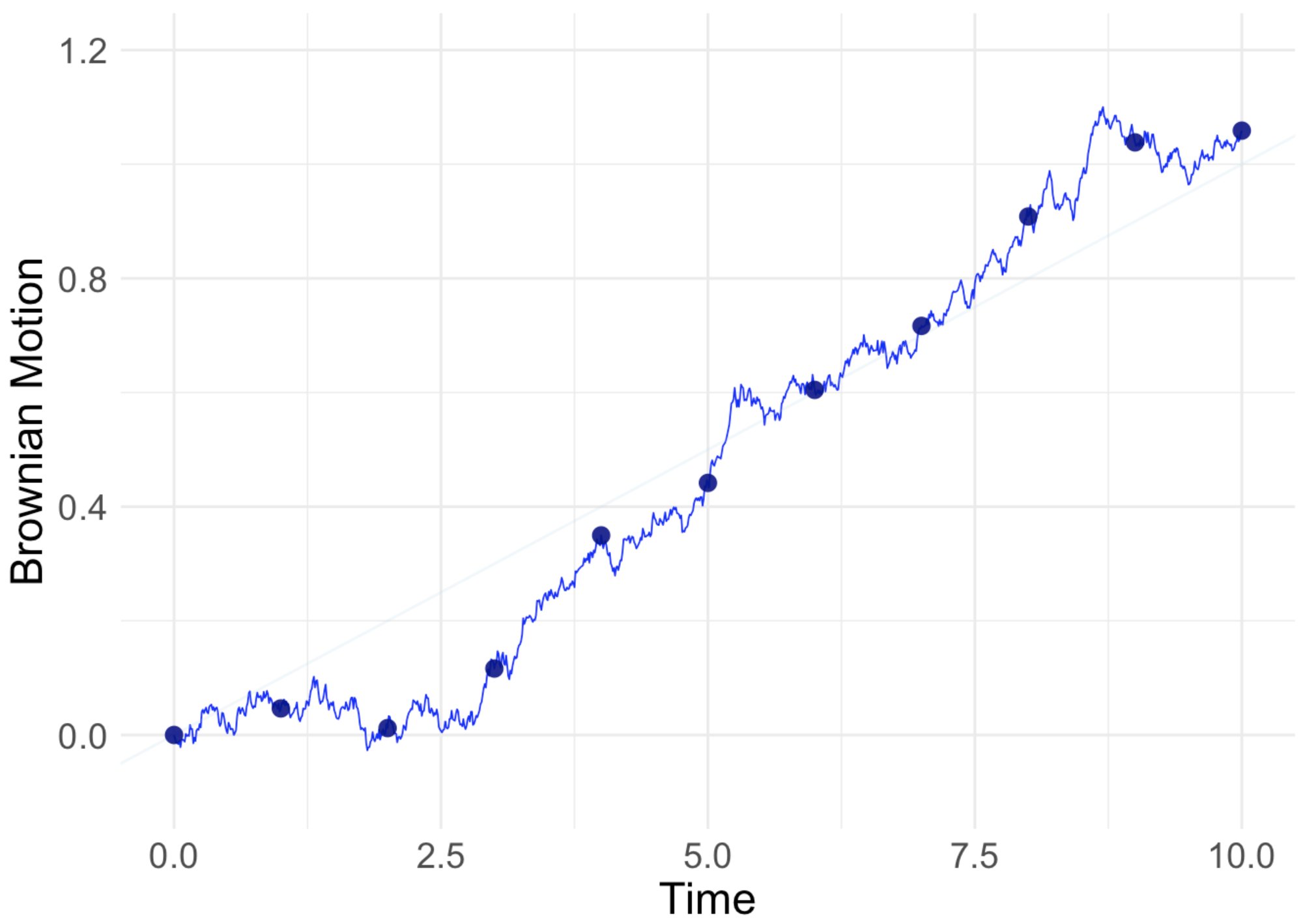} }}
    \qquad
    \subfloat[\centering]{{\includegraphics[width=0.45\linewidth]{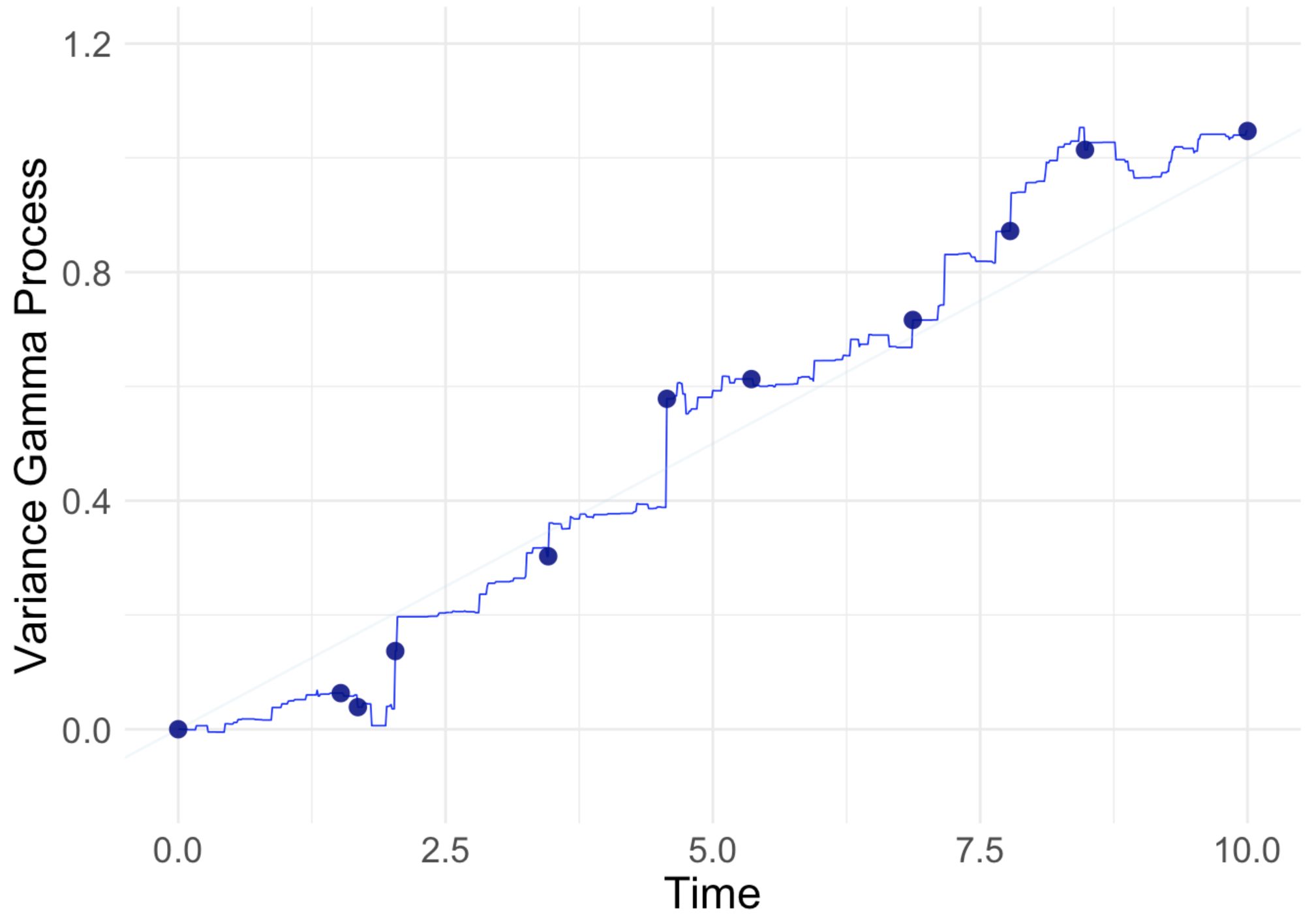} }}
    \caption{Realised components of a $VG$ process: $VG(t;\theta = 0.1, \sigma = 0.1, \nu = 0.3)$}
    \label{VGcomponents}
\end{figure}

Using this idea we can obtain the density function of the process at time $t$ by conditioning on a realisation of the Gamma time-subordinator (giving  a Brownian motion density), and integrating against the Gamma density. Similar to (\ref{General VG Density}), we have
\begin{equation}
    f_{X_t}(x) = \int\phi\left(\dfrac{x-\theta g}{\sigma\sqrt{g}}\right)f_{\gamma_t}(g)dg
    = \int_{0}^\infty\dfrac{1}{\sqrt{2\pi\sigma^2g}}\exp\left(-\frac{(x-\theta g)^2}{2\sigma^2g}\right)\dfrac{g^{\frac{t}{\nu}-1}\exp\left(-\frac{g}{\nu}\right)}{\nu^{\frac{t}{\nu}}\Gamma\left(\frac{t}{\nu}\right)}dg.
\end{equation}
We may also obtain the characteristic function for the process. Again, we condition on realisation of the Gamma time-subordinator with a similar method to the proof of Prop. (\ref{VGcf}), using the characteristic functions in (\ref{bmcf}) and (\ref{gammatcf}). We then find
\begin{equation}\label{VGc.f.}
    \psi_{X_t}(u) = \left(\dfrac{1}{1-i\theta\nu u + \frac{\sigma^2}{2}\nu u^2}\right)^\frac{t}{\nu}.
\end{equation}
\begin{remark}\label{Normalnest}
    Taking the limit as $\nu\downarrow0$ in the $VG$ process, we obtain weak convergence to a Brownian motion $b(t;\theta,\sigma)$ as in (\ref{Brownian motion}) - this is shown in Appendix \ref{BStoVGAppendix} using the characteristic function above. This should be expected, as the Gamma process subordinating the time becomes degenerate with $0$ variance about the mean $t$, and there is no time-change for the Brownian motion. This shows that the $VG$ model still contains Brownian motion as a sub-model - this is a particularly useful property if one employs a hypothesis test which nests a Brownian motion model null inside a more general $VG$ model alternative.
\end{remark}
One could compute the central moments of the $VG$ process from the characteristic function (\ref{VGc.f.}); however, employing the form (\ref{thetagammasigmawienergamma}) simplifies calculations. Conditioning on the Gamma time-subordinator as $\gamma_t = g$, so that the conditional $VG$ process is a Brownian motion and we write
\begin{equation}
    X\left(t;\theta, \sigma \big\vert \gamma_t=g\right) = \theta g +\sigma W_g.
\end{equation}
We can then take the expectation over the Wiener process, and then the expectation over $g=\gamma_t$ as in A4 of Madan et al. \cite{MadanCarrChang}). 
We obtain,
\begin{align*}
    &\mathbb{E}\left[X_t\right] = \theta t, \hspace{0.2cm}\\&\mathbb{E}\left[\left(X_t-\mathbb{E}[X_t]\right)^2\right] = (\theta^2\nu+\sigma^2)t, \hspace{0.2cm}\\&\mathbb{E}\left[\left(X_t-\mathbb{E}[X_t]\right)^3\right] = (2\theta^3\nu^2+3\sigma^2\theta\nu)t,\\ &\mathbb{E}\left[\left(X_t-\mathbb{E}[X_t]\right)^4\right] = (3\sigma^4\nu+12\sigma^2\theta^2\nu^2 + 6\theta^4\nu^3)t + (3\sigma^4+6\sigma^2\theta^2\nu+3\theta^4\nu^2)t^2.
\end{align*}
If we consider the case $\theta = 0$, then we have no skewness as the third central moment becomes $0$. Furthermore, the kurtosis would then be given by $3(1+\nu/t)$ so that $\nu$ represents the percentage excess kurtosis over the Brownian motion process over a unit time increment (where over longer time increments the excess kurtosis decreases to $0$ linearly in $t$). 

However, for $\theta\neq0$ the variance of the process is higher than that of the Brownian motion - this increase in average variation stems from additional variance contributions from the random drift component under the time-subordination (being deterministic, the drift component in a Brownian motion does not contribute to the variance). As a result, the Variance-Gamma typically values options higher due to this increase in volatility.

\section{Option Pricing}\label{Section4}
\subsection{Physical process}

We are now suitably placed to begin exploring the dynamics of a stock price under a $VG$ model, and the associated European call option pricing mechanism. Our exploration follows a similar method to that of the Black-Scholes formulation for European call options, now extending geometric Brownian motion assumption to a more general geometric $VG$ process.

Consider a continuous-time economy where a stock, a money market account and European options for the stock for all strikes and maturities are traded in a given time frame. Further suppose that there is a constant continuously compounded interest rate of $r$ with corresponding money market account value of $e^{rt}$ at time $t$ following a unit deposit at time $0$. 

Under a geometric $VG$ model, the price of the underlying asset $(S_t)_{t\geq0}$ is given by $S_t = S_0e^{mt+X(t;\theta,\sigma,\nu)+\omega(t)}$, where $X(t;\theta,\sigma,\nu)$ is a $VG$ process, and $\omega(t) = \frac{t}{\nu}\log\left(1-\nu\theta-\nu{\sigma^2}/{2}\right)$, so that the mean return rate is $e^{mt}$ for some $m\in\mathbb{R}$ (this can be verified by evaluating the expectation of $S_t$) 
Writing the process in this form allows one to easily identify the mean return rate.

We detail the pricing of European options under the $VG$ model. Specifically we use an Equivalent Martingale Measure (EMM) under which the discounted process $(e^{-rt}S_t)_{t\geq0}$ is a Martingale quantity - this encompasses an appropriate change of measure from the physical measure to a risk-neutral measure outlined through the Esscher transform method below (first developed by Frederik Esscher in 1932 \cite{Esscher}). By finding such a measure, we use the risk-neutral valuation principle to find the price of European call option is given by the expected payoff of the option with respect to the process under the risk-neutral measure.

\subsection{Esscher transform}
We first provide a definition of the Esscher transform (as presented in section 2 of Gerber and Shiu \cite{Gerber1995OptionPB}).
\begin{definition}
    Let $f(x,t)$ be the density for a Lévy process $(X_t)_{t\geq0}$ at time $t$. One defines the \textit{Esscher transform} of $f(x,t)$ as
\begin{equation}
    \hat{f}(x,t,h) = \dfrac{e^{hx}f(x,t)}{M(h,t)},
\end{equation}
where $M(\cdot,t)$ is the moment generating function of $X_t$ at time $t$.
\end{definition}
\begin{remark}
    The Esscher transform of a probability density is itself a new probability density (as the m.g.f. in the denominator normalises the function to integrate to 1) and this associated Esscher transformed measure is equivalent (mutually absolutely continuous) to the original. Denoting the c.d.f. of $X_t$, $F_X(x,t)$,
    \begin{equation}
        F_{X}(x,t) = \int_{-\infty}^xf(\xi,t)d\xi,
    \end{equation}
    where $f(x,t)$ is the density of $X_t$. Then the Esscher transformed variable denoted $X_t^h$, has a cumulative function given by,
    \begin{equation}
        \hat{F}(x,t,h) = \int_{-\infty}^x \hat{f}(\xi,t,h) d\xi =  \int_{-\infty}^x \dfrac{e^{h\xi}f(\xi,t)}{M(h,t)} d\xi.
    \end{equation}
\end{remark}
From this, we can obtain an EMM for the process $X_t$ \footnote{Note that an EMM is not necessarily unique, and indeed when $(X_t)_{t\geq0}$ follows a $VG$ process we do not have uniqueness. Indeed for geometric Lévy processes, only Brownian motion and compensated Poisson processes have a unique EMM as shown in Théorème 1 of \cite{Yor1978}. This means that there are a range of no-arbitrage prices for the $VG$ process.

This is commented on by Andrusiv and Engelbert \cite{esscherentropy} who show that for general Lévy financial markets, the minimal entropy martingale measure coincides with the Esscher martingale measure, and argue that it is an optimal EMM for various utility maximisation problems.
}. 
Letting the physical measure be $\mathbb{P}$ and the Esscher transformed measure be $\mathbb{Q}^h$, we may write the corresponding Radon-Nikodym derivative as,
\begin{equation}
    \dfrac{d\mathbb{Q}^h}{d\mathbb{P}}= \frac{e^{hX_t}}{M(h,t)}
\end{equation}
where we need to choose $h$ appropriately to ensure $\mathbb{Q}^h$ is a Martingale measure for the discounted process $(e^{-rt}S_t)_{t\geq0}$ where $S_t=S_0e^{X_t}$. 
\begin{lemma}\label{hstarlemma}
    The value of $h = h^*$ which makes the discounted process a Martingale quantity under the corresponding transformed measure is determined implicitly through
    \begin{equation}\label{hstar}
    e^{rt}=\dfrac{M(h^*+1,t)}{M(h^*,t)}.
\end{equation}
\end{lemma}
\begin{proof}
    Consider the process $(e^{-rt}S_t)_{t\geq0}$ with filtration $(\mathcal{F}_t)_{t\geq0}$. We seek some $h^*$ such that for the measure $\mathbb{Q}:=\mathbb{Q}^{h^*}$, we satisfy the Martingale equation, $\mathbb{E}^\mathbb{Q}\left[e^{-rt}S_t\hspace{0.1cm}\big\vert \mathcal{F}_0\right]=S_0$.
We must then satisfy
\begin{align*}
    S_0=\mathbb{E}^\mathbb{Q}\left[e^{-rt}S_t\hspace{0.1cm}\big\vert \mathcal{F}_0\right]=\mathbb{E}^\mathbb{P}\left[\dfrac{d\mathbb{Q}}{d\mathbb{P}}e^{-rt}S_t\hspace{0.1cm}\Big\vert \mathcal{F}_0\right]=
    S_0e^{-rt}\mathbb{E}^\mathbb{P}\left[\dfrac{e^{(1+h^*)X_t}}{M(h^*,t)}\hspace{0.1cm}\Big\vert \mathcal{F}_0\right]=S_0e^{-rt}\dfrac{M(h^*+1,t)}{M(h^*,t)},
\end{align*}
so that $h^*$ is indeed determined implicitly through the relation (\ref{hstar}).
\end{proof}
We employ the transform in the context of our $VG$ process. Let $M(\cdot,t)$ be the moment generating function of the $VG$ model $X(t;\theta,\sigma,\nu)$. Then, from integration,
\begin{equation}
    M(h,t) = \left(\dfrac{1}{1-\nu\theta h-\nu\frac{\sigma^2}{2}h^2}\right)^\frac{t}{\nu}, \hspace{0.5cm} h_1<h<h_2
\end{equation}
where (considering the region for which the quadratic denominator is non-zero)
\begin{equation}
    h_1 = -\frac{\theta}{\sigma^2} - \sqrt{\frac{\theta^2}{\sigma^4}+\frac{2}{\nu\sigma^2}}\hspace{0.25cm},\hspace{0.25cm}h_2 = -\frac{\theta}{\sigma^2} + \sqrt{\frac{\theta^2}{\sigma^4}+\frac{2}{\nu\sigma^2}}
\end{equation}
\begin{remark}\label{hunique}
    The existence and uniqueness of such a $h^*$ in the $VG$ case is shown in the Appendix \ref{huniqueappendix}.
    As relation (\ref{hstar}) for the $VG$ process, 
    \begin{equation}
        e^{rt} = \left(\dfrac{1-\nu\theta h^* - \nu\frac{\sigma^2}{2}{h^*}^2}{1-\nu\theta (h*+1) - \nu\frac{\sigma^2}{2}(h*+1)^2}\right)^{\frac{t}{\nu}}
    \end{equation}
    can be expressed as a quadratic in $h^*$ (by multiplying the numerator on both sides), one can compute a formula for $h^*$ in terms of the other parameters; however, the resulting formula is very complicated and we omit this here. 
    
    In the case of a Brownian motion $b(t;\theta,\sigma)$ (or considering a degenerate $VG$ process with $\nu=0$) one finds $h^* = (r-\theta)/\sigma^2-1/2$. Furthermore, as all contingent claims can be replicated in the Black-Scholes market, there is a unique EMM for the Brownian motion model, and
    the Esscher transformed measure with this $h^{*}$ gives this unique EMM exactly when we restrict to this case.
\end{remark}
Consider the set of measures $\{\mathbb{Q}^h:h\in(h_1,h_2)\}$ for the $VG$ process. Computing the m.g.f. of the $VG$ model $M(\cdot,t,h)$ with respect to the measure $\mathbb{Q}^h$,
\begin{align}
    M(z,t,h) &= \mathbb{E}^{\mathbb{Q}^h}\left[e^{zX_t}\right]=\int e^{zX_t}\mathbb{Q}^h(dX_t) = \int_{-\infty}^\infty e^{zx}\hat{f}(x,t,h)dx = \int_{-\infty}^\infty \dfrac{e^{(h+z)x}f(x,t)}{M(h,t)}dx \\ & = \dfrac{M(h+z,t)}{M(h,t)} = \left(\dfrac{1-\nu\theta h - \nu\frac{\sigma^2}{2}h^2}{1-\nu\theta (h+z) - \nu\frac{\sigma^2}{2}(h+z)^2}\right)^{\frac{t}{\nu}}\label{Essmgf}.
\end{align}
Having established the moment generating function of the process under $\mathbb{Q}^h$, we may give the following Proposition.
\begin{prop}\label{EsscherVGVG}
    The Esscher transform of the $VG$ process $X(t;\sigma, \nu, \theta)$ is also a $VG$ process $\Tilde{X}(t;\sigma, \Tilde{\nu},\Tilde{\theta})$ where
    \begin{equation}
        \Tilde{\theta} = \theta + h\sigma^2,\hspace{0.5cm} \Tilde{\nu} = \dfrac{\nu}{1-\nu\theta h - \nu\frac{\sigma^2}{2}h^2}.
    \end{equation}
\end{prop}
The values are obtained upon rearranging (\ref{Essmgf}) (effectively dividing the numerator and denominator in (\ref{Essmgf}) by the numerator) to find
\begin{equation}\label{numdennum}
    M(z,t,h) = \left(\dfrac{1}{1-\Tilde{\nu}\Tilde{\theta}z-\Tilde{\nu}\frac{\sigma^2}{2}z^2}\right)^\frac{t}{\nu},
\end{equation}
where $\Tilde{\theta},\Tilde{\nu}$ are as above, which is indeed the moment generating function of a $VG$ process with the new parameters. The full calculation can be found in the Appendix \ref{EsscherVGVGappendix}.
\begin{remark}\label{rnparams}
    For the case $h=h^*$, one labels the transformed parameters the risk-neutral parameters. In the $VG$ model, we have three risk-neutral parameters as above; however, for a Brownian motion $b(t;\theta,\sigma)$ the risk neutral parameters are independent of $\theta$, with $\Tilde{\theta} = r-\sigma^2/2$ and $\Tilde{\sigma} = \sigma$ is unchanged. 

    Intuitively this is to be expected - a deterministic drift $\theta t$ should have no affect on the risk-neutral measure (one can formalise this notion with a portfolio replication argument). The fact that there is only one parameter governing the risk-neutral density for a Brownian motion means that the option pricing dynamics are constrained to only vary with one variable $\sigma$ which, as shown in the empirical results of Madan et al. \cite{MadanCarrChang}, proves restrictive. 
    
    However, the $VG$ risk-neutral measure maintains a drift term. In particular, the stochastic time change of the $VG$ process affects the drift term meaning that the drift is now also stochastic - a faster time change causing greater drift, and slower time change causing weaker drift. Empirically, it appears the inclusion of the drift term in the $VG$ model this gives significant reduction in mean squared error over Black-Scholes when comparing the prices S$\&$P500 European options on the CBOE with each model.
\end{remark}

\subsection{Variance Gamma European call option price}
Having found our desired EMM $\mathbb{Q}$, we explore the call option pricing mechanism for a $VG$ process, following a similar analysis to Nzokem (2023) \cite{Nzokem}.
\begin{thm}
    Suppose we model the value of an underlying asset through a process $(S_t)_{t\geq0}$ with filtration $(\mathcal{F}_t)_{t\geq0}$. From the risk-neutral valuation principle, given a spot price $S_0$ at time $0$, the price of a European call option $\mathcal{C}(S_0; K,t)$ for a strike price $K$ and maturity $t$ is given by
    \begin{equation}\label{call}
   \mathcal{C}(S_0; K,t) = e^{-rt}\mathbb{E}^\mathbb{Q}\left[\left(S_t-K\right)_+\hspace{0.1cm}\big\vert\mathcal{F}_0\right],
    \end{equation}
    where the expectation is taken under the risk neutral measure $\mathbb{Q}$. If the price $S_t$ follows a geometric $VG$ process such that $S_t=S_0e^{X_t}$ for $X_t = X(t;\theta,\sigma,\nu)$, the European call option price is given by
    \begin{equation}\label{Calloptionformula:)}
        \mathcal{C}(S_0; K,t) = S_0\left[1-\hat{F}\left(\log\left(\frac{K}{S_0}\right),t,h^*+1\right)\right] - Ke^{-rt}\left[1-\hat{F}\left(\log\left(\frac{K}{S_0}\right),t,h^*\right)\right],
    \end{equation}
    where $\hat{F}{(x,t,h)}$ is the cumulative function of the Esscher transformed $VG(t;\theta,\sigma,\nu)$ process,
    \begin{equation}
        \hat{F}(x,t,h) = \int_{-\infty}^x \hat{f}(\xi,t,h) d\xi.
    \end{equation}
\end{thm}
\begin{proof}
    From equation (\ref{call}) we may write 
    \begin{equation}
        \mathcal{C}(S_0; K,t) = e^{-rt}\mathbb{E}^\mathbb{Q}\left[\left(S_t-K\right)_+\hspace{0.1cm}\big\vert\mathcal{F}_0\right] = e^{-rt}\mathbb{E}^\mathbb{Q}\left[\left(S_0e^{X_t}-K\right)_+\hspace{0.1cm}\big\vert\mathcal{F}_0\right].
    \end{equation}
    Our risk-neutral measure $\mathbb{Q}$ has a density $\hat{f}(x,t,h^*)$ and we write
    \begin{equation}
        \mathcal{C}(S_0; K,t) = e^{-rt}\int (S_0e^{X_t}-K)_+\mathbb{Q}(dX_t) = e^{-rt}\int_{-\infty}^\infty (S_0e^{\xi}-K)_+\hat{f}(\xi,t,h^*)d\xi.
    \end{equation}
    Considering the integrand is non-zero for $\xi\in\left(\log k,\infty\right)$ where $k=K/S_0$, we may split the integral.
    \begin{equation}
        e^{-rt}\int_{-\infty}^\infty (S_0e^{\xi}-K)_+\hat{f}(\xi,t,h^*)d\xi = e^{-rt}S_0\int_{\log k}^\infty e^{\xi}\hat{f}(\xi,t,h^*)d\xi-Ke^{-rt}\int_{\log k}^\infty \hat{f}(\xi,t,h^*)d\xi.
    \end{equation}
    The second integral term is already in the desired form, but the first integrates $e^\xi$ against the risk-neutral density. We employ our results on the moment generating function and the implicit relation for $h^*$ from Lemma (\ref{hstarlemma}). From the definition of the Esscher transform $\hat{f}$,
    \begin{equation}
        \int_{\log k}^\infty e^{\xi}\hat{f}(\xi,t,h^*)d\xi = \int_{\log k}^\infty e^{\xi}\dfrac{e^{h^*\xi}f(\xi,t)}{M(h^*,t)}d\xi = \dfrac{M(h^*+1,t)}{M(h^*,t)}\int_{\log k}^\infty \dfrac{e^{(h^*+1)\xi}f(\xi,t)}{M(h^*+1,t)}d\xi.
    \end{equation}
    Recall from the implicit definition for $h^*$ from (\ref{hstar}) that the factor outside the integral is equal to $e^{rt}$. The integrand is now exactly $\hat{f}(\xi,t,h^*+1)$, and so
    \begin{equation}
        \int_{\log k}^\infty e^{\xi}\hat{f}(\xi,t,h^*)d\xi = e^{rt}\int_{\log k}^\infty \hat{f}(\xi,t,h^*+1)d\xi.
    \end{equation}
    Then indeed
    \begin{equation}
        \mathcal{C}(S_0;K,t) = S_0\int_{\log k}^\infty \hat{f}(\xi,t,h^*+1)d\xi-Ke^{-rt}\int_{\log k}^\infty \hat{f}(\xi,t,h^*)d\xi,
    \end{equation}
    which yields the result upon integration.
\end{proof}
\begin{remark}
    The formula demonstrates strong resemblance with the equivalent Black-Scholes European call option pricing formula. 
    The first term represents the present value $S_0$ of holding the stock if the option is exercised, weighted against the probability that the future value $S_t$ finishes in the money. The second term represents the cost $Ke^{-rt}$ of exercising the option weighted against the probability that the future value $S_t$ exceeds the strike. 
\end{remark}

\section{Comparison with Black-Scholes}\label{Section5}
Having obtained an option pricing mechanism, we explore properties of the $VG$ call option price and compare with the well-documented Black-Scholes approach

Suppose that the price of an underlying asset follows a geometric $VG$ process, with a fixed strike price of $100$ for a range of European call options with a fixed 2 week maturity. The corresponding call option price for each spot price of the underlying is given in Figure (\ref{RNVGplot}) below, varying $\sigma_{VG}$.

\begin{figure}[h!]
    \centering
    \includegraphics[width=0.5\linewidth]{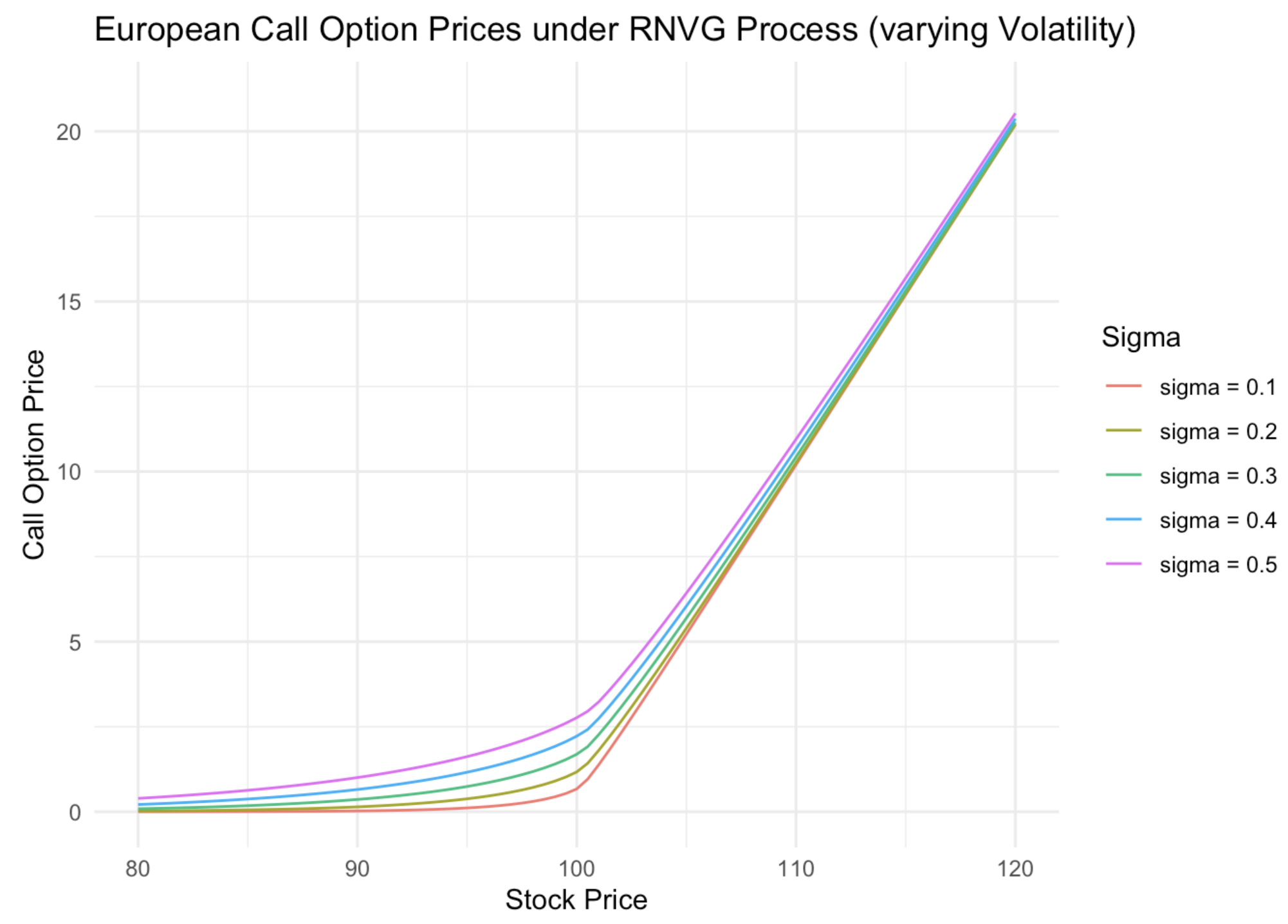}
    \caption{Option prices under a $VG$ model. $r = 0.05, \theta = -0.1, \nu = 0.2$, time to maturity 2 weeks.}
    \label{RNVGplot}
\end{figure}

\begin{remark}
    We obtain the familiar hockey stick shape for the call option price with greater call option values for greater volatility (caused by a greater $\sigma$ parameter). This agrees with our expectation - for out of the money spot prices, the probability that the option will be exercised is low and the corresponding option has lower value. The reverse applies for in the money spot prices where there is greater probability that the option will be exercised. A higher volatility (through greater $\sigma$) corresponds to greater expected fluctuation in the stock price and so greater expected returns from the option.
\end{remark}

We now compare with the Black-Scholes sub-model price for the equivalent option in Figure (\ref{VG-BS}) (where the $VG$ parameters are the same while varying $\sigma_{VG}$). The Black-Scholes model volatility $\sigma_{BS}$ is chosen to match the volatilities of the corresponding $VG$ model using the moment calculations of section \ref{SectionVGProcess}, so that $\theta_{VG}^2\nu_{VG}+\sigma^2_{VG}=\sigma^2_{BS}$. Recall $\theta_{VG}$ represents the average drift of the $VG$ process, while $\nu_{VG}$ represents the variance of the Gamma process subordinating the Brownian motion for the $VG$ process.

Figure \ref{VG-BS}(a) demonstrates that the Variance Gamma and Black-Scholes models tend to generally price options in a similar manner, following the typical hockey stick shape. However, Figure \ref{VG-BS}(b) demonstrates a clear trend to the difference in the option prices between the two models. For spot prices further in or out of the money, the Variance Gamma option price is significantly higher than the Black-Scholes, whereas for spot prices close to at the money the Black-Scholes option price is much higher, with increasing volatility exaggerating the difference in both cases.
\begin{figure}[h!]
    \centering
    \subfloat[\centering]{{\includegraphics[width=0.45\linewidth]{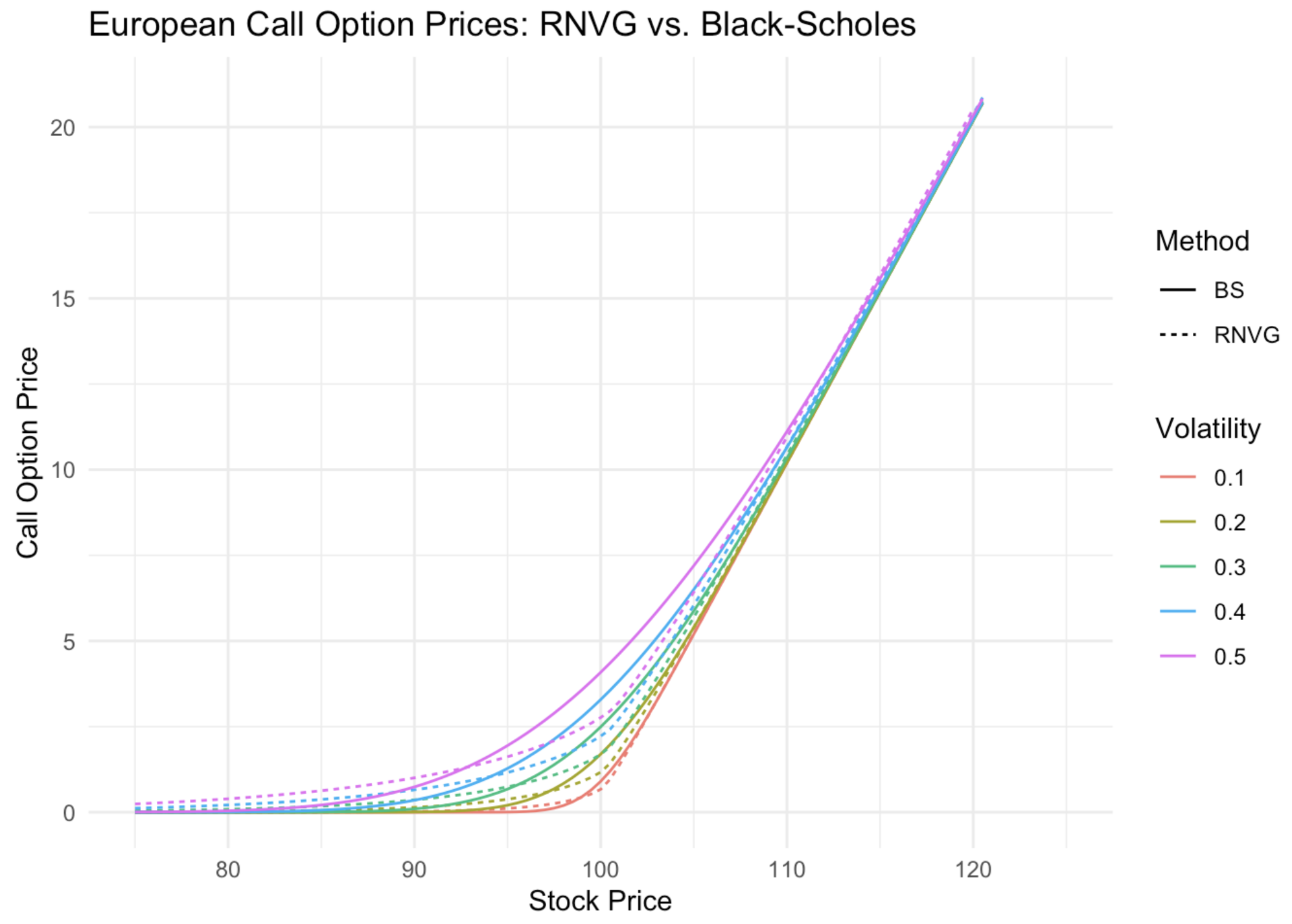} }}
    \qquad
    \subfloat[\centering]{{\includegraphics[width=0.45\linewidth]{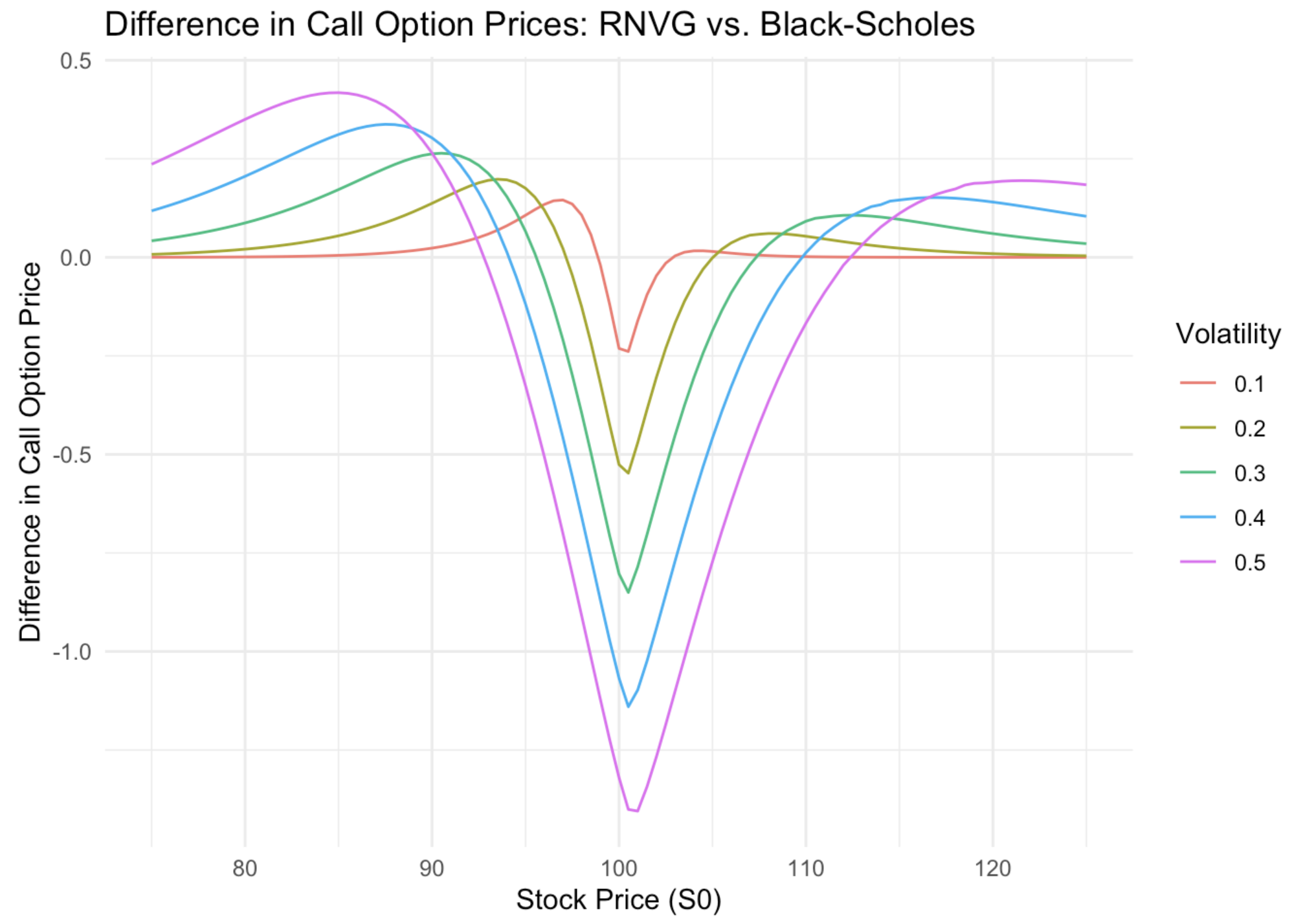} }}
    \caption{Variance Gamma - Black-Scholes difference in call option price.}
    \label{VG-BS}
\end{figure}

\begin{remark}
    Once again this is expected - as explained in the previous sections, the more general Variance Gamma distribution is leptokurtic exhibiting greater kurtosis than the Gaussian with sharper peaks about the median and heavier tails.
    
    For spot prices close to at the money, one would not expect the price of the underlying asset to deviate as much under a Variance Gamma model compared to Black-Scholes - the sharper peak around the centre of the density function attributes greater probability to the price deviating minimally. However, the Variance Gamma model assigns greater probability to a larger price deviation than Black-Scholes thereby increasing the value of the corresponding option for spot prices far in or out of the money.
\end{remark} 

\begin{figure}[h!]
    \centering
    \includegraphics[width=0.5\linewidth]{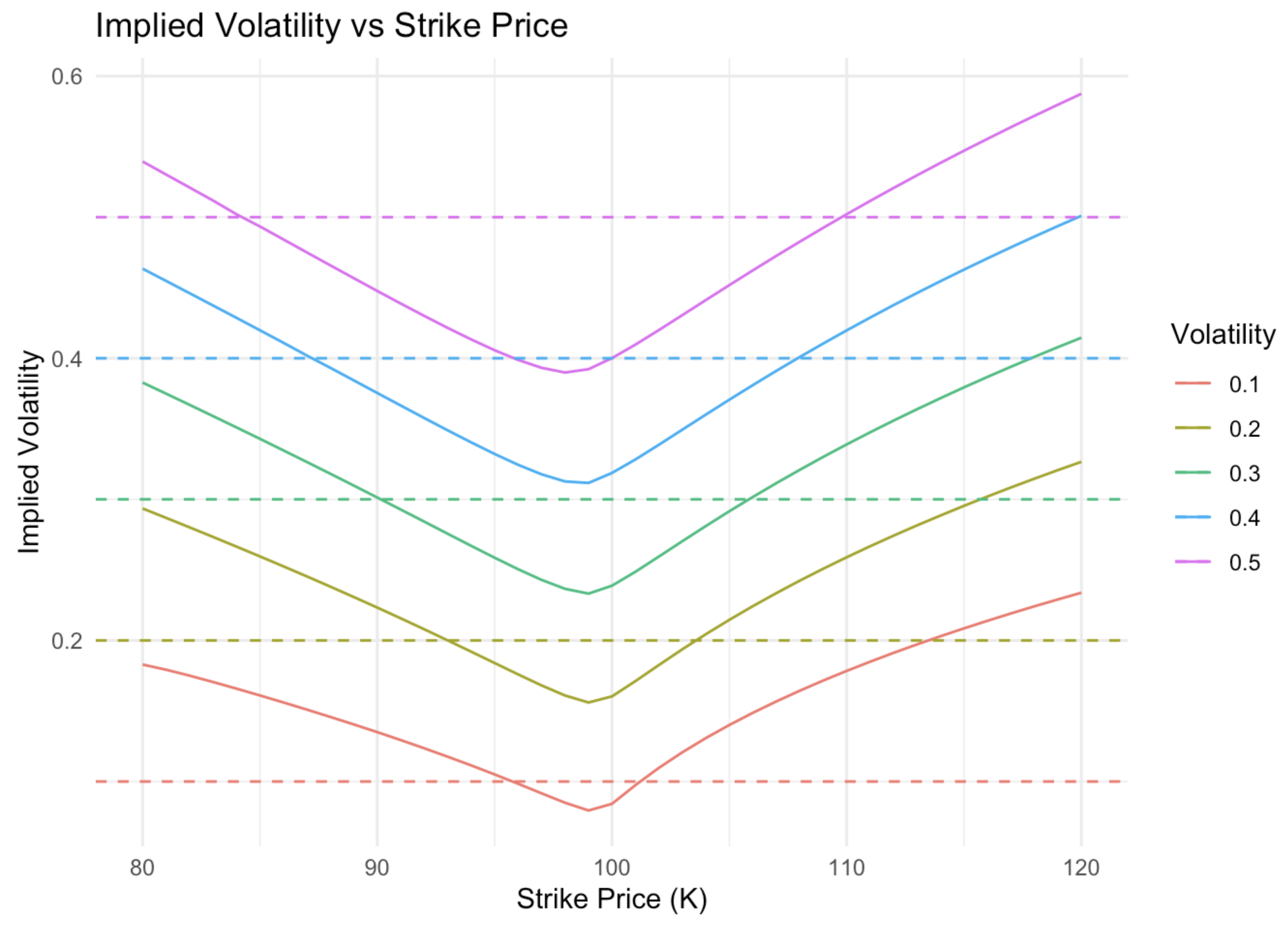}
    \caption{Implied Black-Scholes volatility from Variance Gamma simulated prices}
    \label{Simimpliedvolatility}
\end{figure}

From this result, one may question whether fitting a Black-Scholes implied volatility to prices simulated from a general Variance Gamma model with $\nu$ significantly greater than $0$ will lead to a volatility smile, as the phenomenon is often attributed to greater kurtosis in  log-returns data. Figure (\ref{Simimpliedvolatility}) shows that this is exactly the case - we invert the Black Scholes formula at each strike for $VG$ calculated prices to calculate the implied volatility. The spot price for Figure (\ref{Simimpliedvolatility}) is fixed at 100 units, and the other Variance Gamma parameters are as above (notably $\nu = 0.2$ significantly greater than $0$). We also observe that increasing the volatility of the Variance Gamma model exaggerates the smile.

Empirically, it does appear that the $VG$ model demonstrates significant improvement over the nested Black-Scholes when fitted to options data. This is shown in section 5 of Madan et al. \cite{MadanCarrChang} where the authors fit both $VG$ and Black-Scholes models on a weekly basis for 143 weeks of S$\&$P500 data, rejecting a Black-Scholes null hypothesis at the 5$\%$ level for $93.7\%$ of weeks, and at the 1$\%$ level for $91.6\%$ of weeks.

It also appears that the $VG$ model's additional parameters are capable of flattening the empirical implied volatility smile - this is shown in section 6 of Madan et al. \cite{MadanCarrChang} which deals with the pricing performance of the $VG$ model against the nested Black-Scholes. The authors employ regression analysis through orthogonality tests for both models - these test whether there are any consistent patterns in pricing errors due to variables such as moneyness, maturity, and the interest rate. The authors find that for the Black-Scholes and symmetric $VG$ models, the pricing error has a high moneyness dependency for the options in the data (the symmetric $VG$ does still perform better than Black-Scholes); however, the pricing error for the general $VG$ model (with drift) is mostly free of this dependency.

\section{Empirical Results}
\subsection{Data for the study}
For our study we have used S$\&$P500 futures traded at the Chicago Board Options Exchange (CBOE). The S$\&$P 500 futures options we consider are European call options traded with a 7-day maturity between August 2022 and August 2023.

The data for the study were obtained from the Wharton Research Data Services and include all transaction option prices from August 31, 2022 to August 31 2023. Closing prices on index futures and the level of the spot index were also available from the Federal Reserve Economic Data. All options contracts were viewed as written on the underlying spot index. There were 46135 options prices over the 1 year period considered in the analysis.

\subsection{Daily Log Returns}
We begin the analysis by estimating the parameter values of the statistical densities for daily log returns over the period above for the $VG$ model and the nested Gaussian sub-model. The data employed were the 262 daily observations of log spot price of the S$\&$P500 Index covering the period. For the stock price dynamics under the $VG$ model, we employ the density in equation \ref{VG density} while the stock price dynamics underlying the Gaussian sub-model is of course a Gaussian density with mean $\theta$ and standard deviation $\sigma$. Recall from Remark (\ref{Normalnest}) that this Gaussian sub-model is nested in the $VG$ model through the degenerate case where $\nu=0$ which permits us to perform an approximate likelihood ratio test on the results. We employ maximum likelihood estimation for all parameters and present the estimated parameter values in Table (\ref{Daily table}).
\begin{table}[h!]
\begin{center}
\begin{tabular}{c c c}
 \hline
 Estimated Parameter & Gaussian Model & $VG$ Model \\
 \hline\hline
 $\theta$ & 0.0005925507 & -0.001323872 \\
 $\sigma$ & 0.01141282 & 0.01201207 \\ 
 $\nu$ & $\cdot$ & 0.02942378 \\
 $\log \mathcal{L}$ & 1004.44275 & 1012.215\\
 NOBS & 252 & 252 \\
 \hline
\end{tabular}
\caption{Fitted S$\&$P500 daily log returns}
\label{Daily table}
\end{center}
\end{table}

The estimated $\nu$ parameter of the $VG$ model $\nu = 0.02942378$ appears to be significantly different from zero which suggests we should expect to reject a Gaussian null hypothesis in favour of a $VG$ alternative. Indeed, with the use of Wilks' theorem we can perform an approximate chi-squared test, in that under the null hypothesis, the statistic
\begin{equation}\label{2xLRT}
    2 \left(\log \mathcal{L}_{VG} - \log \mathcal{L}_{\text{Gaussian}}\right)
\end{equation}
is approximately chi-squared distributed with 1 degree of freedom (as the $VG$ model has one additional parameter). Here the statistic is equal to 15.5445 which strongly rejects\footnote{this gives a p-value less than 0.0001 under the $\chi^2_1$ distribution.} the Gaussian hypothesis in favour of the $VG$ model.

\subsection{Option pricing}
For our theoretical $VG$ option pricing mechanism, we made use of the Esscher transform to obtain the risk-neutral measure for the $VG$ model. We employ the density of this measure for a similar parameter estimation, this time estimating the parameters of both the risk-neutral Black-Scholes density and risk-neutral $VG$ density used to appropriately price options. We fit both densities to the option price data on a weekly basis via a maximum-likelihood method.

We present the average estimated parameter values of the risk neutral densities, along with
their standard deviation across 52 weeks, as well as the minimum and maximum estimated values over the period in Table (\ref{optionsstats}). We also include the log-likelihood statistics for each model, again showing the average and standard deviation across the weeks, as well as the minimum and maximum values over the period.
\begin{table}[h!]
    \centering
    \begin{tabular}{c c c c c}
    \hline
    Parameter & Mean & Standard deviation & Minimum & Maximum \\ [1ex]
    \hline
    \multicolumn{5}{l}{\emph{Black Scholes}} \\
    $\sigma$ & 0.18010 & 0.03643 & 0.03128 & 0.28604 \\
    $\log\mathcal{L}$ & 0.00422 & 0.00239 & 0.00182 & 0.01723 \\[2ex]
    \multicolumn{5}{l}{\emph{Variance Gamma}} \\
    $\sigma$ & 0.1793376 & 0.0505 & 0.02983 & 0.30049 \\
    $\theta$ & 0.03015743 & 0.51320 & -1.57233 & 1.65062 \\
    $\nu$ & 0.01231924 & 0.01698 & 0.00528 & 0.02093 \\
    $\log\mathcal{L}$ & 0.00523 & 0.00301 & 0.00191 & 0.01799
\end{tabular}
    \caption{Weekly parameter estimation of Black-Scholes and $VG$ model S$\&$P500 options}
    \label{optionsstats}
\end{table}

For each week we construct an approximate chi-squared test in the same form as equation (\ref{2xLRT}), though this time under a Black-Scholes null hypothesis the statistic in (\ref{2xLRT}) now follows a $\chi^2_2$ distribution, as the risk-neutral $VG$ density has 2 further degrees of freedom ($\theta,\nu$) over the risk-neutral Black-Scholes. From the data, the Black Scholes model is rejected at the 5$\%$ level\footnote{noting that the $\chi^2_2$ has a critical value of $0.00011$ at the $5\%$ level, and $0.02010$ at the $1\%$ level} for $48$ out of the $52$ weeks ($92.3\%$) in favour of the $VG$ model, and rejected at the $1\%$ level for $45$ out of the $52$ weeks ($86.5\%$) in favour of the $VG$ model.

\section{Summary}\label{Section6}
This paper explores the concept of stochastic volatility in modelling stock-price dynamics, first presenting results on the Laplace distribution as a Gaussian variance-mixture and generalising this idea to find the powerful variance gamma model as a Gamma time-subordinated Brownian motion to price European call options.

We propose a more efficient form of volatility estimation when one assumes that the log-returns from an underlying follow a Laplace distribution. Notably, the mean absolute absolute deviation about the median proves a more efficient estimator for the volatility parameter.

The Laplace distribution is generalised to consider the $VG$ distribution as a Gaussian with gamma distributed variance. We then characterise the $VG$ process as a time-subordinated Brownian motion with a Gamma process subordinator.

A general European call option pricing mechanism is formulated via the use of the Esscher transform to find a corresponding EMM measure, and we investigate the properties of the EMM measure for the $VG$ model. This allows us to compute the European call option price under the $VG$ model.

The theoretical differences between the more general $VG$ model and the nested Black-Scholes are explored - the extra degrees of freedom for the $VG$ allow the more general model to better capture the excess kurtosis of returns data, and provide a possible explanation for the presence of implied volatility smiles when fitting the Black-Scholes model. Referring to works comparing the empirical performance of the $VG$ model against the Black-Scholes both in modelling daily log returns and the appropriate pricing of options for various markets, a majority of studies find that the Black-Scholes model is strongly rejected in favour of the more general $VG$ model.

\section*{Acknowledgements}

We would like to acknowledge the Imperial-MIT International Research Opportunities Programme without which this research collaboration would not have been possible. Rohan Shenoy would like to acknowledge the financial support provided by the Mathematics Department at Imperial College London, the Imperial International Relations Office, and the UK Government's Turing Scheme.

\bibliographystyle{unsrt}
\bibliography{Bibliography.bib}

\section{Appendix}

\subsection{Proof of Lemma \ref{thetaestimate}}\label{thetaestimateappendix}
\begin{proof}
    Assuming $\theta$ is given, we differentiate the log-likelihood with respect to $s$,
    \begin{equation}
        \frac{\partial}{\partial s}l(x_1,\dots,x_n; s, \theta) = -\frac{n}{ s}+\frac{1}{ s^2}\sum_{i=1}^{n}|x_{i}-\theta|.
    \end{equation}
    Setting this equal to $0$ we do find $\hat{ s}$ is given by (\ref{sigmahat}). Finding the variance of the estimator,
    \begin{equation}
        \text{Var }\hat{s} = \text{Var}\left[\frac{1}{n}\sum_{i=1}^n|x_i-{\theta}|\right] = \frac{1}{n}\left(\mathbb{E}\left[|X_i-\theta|^2\right] - \left(\mathbb{E}\left[|X_i-\theta|\right]\right)^2 \right)\\
        = \frac{1}{n}(2 s^2- s^2) = \frac{ s^2
        }{n}.
    \end{equation}
    This corresponds with the Cramér-Rao lower bound - indeed, calculating the Fisher information
    \begin{align*}
        \mathbb{I}_n(f) = -n\mathbb{E}\left[\frac{\partial^2}{\partial  s^2}\log f(x)\right] = -n\mathbb{E}\left[\frac{1}{ s^2}-\frac{2|X-\theta|}{ s^3}\right] = -n\left(\frac{1}{ s^2} - \frac{2 s}{ s^3}\right) = \frac{n}{ s^2},
    \end{align*}
    (where we have used (\ref{centralabsolutemoment}) for the central absolute moment) giving the result on taking inverses.
\end{proof}

\subsection{Proof of Lemma \ref{sigmaestimate}}\label{sigmaestimateappendix}
\begin{proof}
    From the strong law of large numbers, $\frac{1}{n}\sum_{i=1}^n(X_i-\theta)^2\overset{p}{\to}\mathbb{E}\left((X-\theta)^2\right) = 2\sigma^2$
    so that by a continuity theorem,
    \begin{equation}
        \sqrt{\frac{1}{2}\left(\frac{1}{n}\sum_{i=1}^n(X_i-\theta)^2\right)}\overset{p}{\to}\sqrt{\frac{1}{2}\left(2\sigma^2\right)} = \sigma
    \end{equation}
    demonstrating the consistency. To establish asymptotic normality, we first note
    \begin{equation}
        \mathbb{E}\left[(X-\theta)^2\right] = 2\sigma^2\hspace{0.5cm},\hspace{0.5cm}\text{Var}\left[(X-\theta)^2\right]=\mathbb{E}[X^4] - \mathbb{E}[X^2]^2 = 20\sigma^4,
    \end{equation}  
    where the moments are calculated from (\ref{centralmoments}). Then from the central limit theorem,
    \begin{equation}
        \sqrt{n}(\Tilde{\sigma}-\sigma)\overset{d}{\to}N(0,20\sigma^4).
    \end{equation}
    Applying the delta method for $g(x) = \sqrt{x/2}$ so that $g'(x) = \frac{1}{2\sqrt{2x}}$, 
    \begin{equation}
        \sqrt{n}(\Tilde{\sigma}-\sigma)\overset{d}{\to}N\left(0,g'(2\sigma^2)20\sigma^4\right)\overset{d}{=}N\left(0,\frac{5}{4}\sigma^2\right).
    \end{equation}
\end{proof}

\subsection{Maximum Entropy distributions}\label{Entropyappendix}

For a random variable $X$ with density $f$, one defines the entropy of $X$ as $H(X) = \mathbb{E}\left[-\log{f(X)}\right]$.
The following theorem from Kagan et al. \cite{KaganLinnikRao} provides a framework for finding the density of the maximum entropy distribution, which maximises this expectation.

\begin{thm}\label{entthm}
    Let $X$ be a random variable with density $p(x)$ and support $(a,b)$ (where $a,b$ can be $\pm\infty$). Let $h_1,\dots,h_n$ be integrable functions on $(a,b)$ satisfying for different some constants $g_i$,
    \begin{equation}\label{his}
        \int_{a}^{b}h_i(x)p(x)dx = g_i.
    \end{equation}
    Then the maximum entropy distribution is attained for densities of the form
    \begin{equation}\label{entform}
        p(x) = e^{a_0+a_1h_1(x)+\cdots a_nh_n(x)}
    \end{equation}
    where the above constraints are all satisfied.
\end{thm}
\noindent The result follows from the use of Lagrange multipliers to find $-\log(p) + \sum_{i=1}^na_ih_i \equiv 0$.
The full proof can be found in Kagan et al. \cite{KaganLinnikRao}.

We now see the power of this theorem. One observes that a Gaussian arises as the maximum entropy distribution where the support is the real line, and $h_1 = x, h_2 = x^2$ so that the first and second moment (or equivalently mean and variance) are specified. The exponent in (\ref{entform}) becomes a quadratic which yields the Gaussian density (we require $a_2>0$ but this is easily shown).

But instead specifying the first absolute moment through $h_1(x) = |x|$, so that $\int_\mathbb{R}|x|p(x)dx = g_1$
we find that the density $p(x)$ is of the form $e^{a_0+a_1|x|}$, a centralised Laplace density. So the Laplace distribution is the maximum entropy distribution when the first absolute moment is specified. 

However, it also arises as the maximum entropy distribution of a Gaussian variance-mixture with specified mean, and the specified mean of its random variance $V$. 
If we specify the mean of $V$, then from theorem (\ref{entthm}) we have
$p(v) = e^{a_0+a_1x}$
which is an exponential distribution. As we have seen, a Gaussian variance-mixture distribution with exponential variance is a Laplace distribution, and so the Laplace distribution arises as the maximum entropy distribution of a Gaussian variance-mixture with specified mean, and specified mean of some random variance $V$. 

This provides additional motivation for the initial consideration of the Laplace distribution.

\subsection{Brownian motion as a degenerate Variance Gamma process}\label{BStoVGAppendix}
Consider taking $\nu\downarrow0$ with in (\ref{VGc.f.}) 
\begin{align}
    \lim_{\nu\downarrow0}\psi_{X_t}(u) &= \lim_{\nu\downarrow0}\left(\dfrac{1}{1-i\theta\nu u + \frac{\sigma^2}{2}\nu u^2}\right)^\frac{t}{\nu} = \lim_{\nu\downarrow0}\left({1-\nu\left(i\theta u + \frac{\sigma^2}{2} u^2\right)}\right)^\frac{-t}{\nu} \\
    &= \lim_{R\uparrow\infty}\left({\left({1+\frac{-i\theta u + \frac{\sigma^2}{2} u^2}{R}}\right)^R}\right)^{-t} = \left(\exp\left(-i\theta u + \frac{\sigma^2}{2}u^2\right)\right)^{-t} = \exp\left(i\theta tu -\frac{1}{2}\sigma^2t u^2 \right).
\end{align}
which is indeed the characteristic function of a Brownian motion (\ref{bmcf}) (where we have used the complex exponential limit in the penultimate equality). Clearly this is continuous at $0$, and so we employ Lévy's continuity theorem (Theorem 15, section 14.7 of Fristedt and Gray \cite{Levycontinuity}) to conclude convergence in distribution to the Brownian motion. As discussed, this is intuitive as the Gamma process subordinating the time becomes degenerate with $0$ variance about the mean $t$, and there is no time-change for the Brownian motion in the limit $\nu\downarrow0$.

\subsection{Proof of remark \ref{hunique}}\label{huniqueappendix}
Let $H(h)$ be equal to inverse of the RHS of (\ref{hstar}) without the exponent $t/\nu$ so that
\begin{equation}
    H(h) = \dfrac{1-\nu\theta h-\nu\frac{
    \sigma^2}{2}h^2}{1-\nu\theta (h+1)-\nu\frac{
    \sigma^2}{2}(h+1)^2}.
\end{equation}
Then we find,
\begin{equation}
    \frac{dH}{dh}(h) = \dfrac{\frac{\nu^2\sigma^4}{2}h^2+\theta\nu^2(\frac{\sigma^2}{2}+\theta)h+\theta\nu^2(\frac{\sigma^2}{2}+\theta)+\nu\sigma^2}{\left(1-\nu\theta(h+1)-\frac{\nu\sigma^2}{2}(h+1)^2\right)^2},
\end{equation}
so that
\begin{equation}
    \frac{dH}{dh}(h) > 0, h_1<h<h_2-1,\hspace{1cm} \lim_{h\to h_1}g(h) = 0,\hspace{1cm}\lim_{h\to h_2-1^-}g(h) = +\infty
\end{equation}
so that $H$ is a bijection from $(h_1,h_2-1)$ to $\mathbb{R^+}$ and there is a unique $h^*$ which satisfies (\ref{hstar}).

\subsection{Proof of proposition \ref{EsscherVGVG}}\label{EsscherVGVGappendix}
Consider the inverse of the term inside the bracket in (\ref{Essmgf}) as suggested. Then the denominator becomes
\begin{equation}
    \dfrac{1-\nu\theta (h+z) - \nu\frac{1}{2}\sigma^2(h+z)^2}{1-\nu\theta h - \nu\frac{1}{2}\sigma^2h^2} = \dfrac{1-\nu\theta h-\nu\theta z - \nu \frac{\sigma^2}{2}(h^2+2hz+z^2)}{1-\nu\theta h - \nu\frac{1}{2}\sigma^2h^2}
\end{equation}
which we may expand as
\begin{equation}
    \dfrac{1-\nu\theta h- \nu \frac{\sigma^2}{2}h^2-\nu(\theta+h\sigma^2)z-\nu\frac{\sigma^2}{2}z^2}{1-\nu\theta h - \nu\frac{1}{2}\sigma^2h^2} = 1 - \dfrac{\nu(\theta+h\sigma^2)z - \nu\frac{\sigma^2}{2}z^2}{1-\nu\theta h - \nu\frac{1}{2}\sigma^2h^2}
\end{equation}
rewriting the fraction gives
\begin{equation}
    1 - \dfrac{\nu}{{1-\nu\theta h - \nu\frac{1}{2}\sigma^2h^2}}(\theta+h\sigma^2)z -\dfrac{\nu}{{1-\nu\theta h - \nu\frac{1}{2}\sigma^2h^2}}\left(\frac{\sigma^2}{2}z^2\right)
\end{equation}
so we can conclude by letting $\Tilde{\theta}, \Tilde{\nu}$ be as given in the proposition.

\end{document}